\documentclass[11pt,a4paper]{article}

\usepackage{fullpage}
\usepackage{amsmath,amssymb,amsthm,ascmac}
\usepackage{comment}
\usepackage{enumerate}
\usepackage[cal=pxtx]{mathalpha}
\usepackage{makecell}
\usepackage{mleftright}
\usepackage{graphicx}
\usepackage{setspace}
\usepackage[most]{tcolorbox}
\definecolor{firebrick}{RGB}{178,34,34}
\usepackage[linesnumbered,ruled,noend,nofillcomment]{algorithm2e}
\SetKwProg{Function}{Function}{:}{}
\SetKwInput{Precondition}{Precondition}
\SetKwComment{Comment}{$\triangleright$ }{}
\SetProgSty{textnormal}
\SetCommentSty{TODO}
\usepackage[pdfencoding=auto,psdextra]{hyperref}
\usepackage{cleveref}
\Crefname{algocf}{Algorithm}{Algorithms}
\usepackage[style=trad-alpha]{biblatex}
\newtheorem{theorem}{Theorem}[section]
\newtheorem{lemma}{Lemma}[section]
\newtheorem{proposition}{Proposition}[section]
\newtheorem{corollary}{Corollary}[section]
\newtheorem{example}{Example}[section]
\newtheorem{claim}{Claim}[section]
\theoremstyle{definition}
\newtheorem{definition}{Definition}[section]

\usepackage{booktabs}
\usepackage{mathtools}
\usepackage[section]{placeins}

\newcommand{\UnderlineCase}[1]{\smallskip\noindent\underline{\emph{#1}}}

\newcommand{\eps}{\varepsilon}
\newcommand{\true}{\emph{true}}
\newcommand{\false}{\emph{false}}
\newcommand{\A}{\mathcal{A}}
\newcommand{\B}{\mathcal{B}}
\newcommand{\last}{\mathrm{last}}
\newcommand{\Comp}{{\sf Comp}}
\newcommand{\SendMaxToMin}{{\sf SendMaxToMin}}
\newcommand{\ScalingSendMaxToMin}{{\sf ScalingSendMaxToMin}}
\newcommand{\PROPOneNonPROP}{{\sf PROP1\text{-}nonPROP}}
\newcommand{\EFOneHalfMMS}{{\sf EF1\text{\sf +}1/2\text{\sf -}TPS}}
\newcommand{\EgalSequential}{{\sf Egal\text{-}Sequential}}
\newcommand{\HallMatching}{{\sf HallMatching}}
\newcommand{\ProtectedScalingSendMaxToMin}{{\sf ProtectedScalingSendMaxToMin}}
\newcommand{\TPS}{\mathrm{TPS}}

\title{Witness-Certified Fair Division with Comparison Queries}

\usepackage{authblk}
\author[1]{Tatsuhito Yamagata}
\author[1]{Hanna Sumita}
\affil{Institute of Science Tokyo, Japan}
\date{}

\begin{document}

\maketitle

\begin{abstract}
We study fair division of indivisible goods when agents' valuations are accessed only through ordinal comparisons between bundles, with arbitrary tie-breaking.
In this model, even deciding whether a given allocation is envy-free up to one good (EF1) can be impossible.
This suggests explicit fairness certificates as a natural algorithmic object.
Our main contribution is a certificate-preserving scaling framework, which recursively contracts goods, solves a smaller instance, and expands the solution while repairing an explicit envy-eliminating witness certificate.
For arbitrary identical monotone valuations, this yields a certified EF1 allocation for $n$ agents and $m$ goods using $O(n \log n \log(m/n))$ comparison queries, within an $O(\log n)$ factor of the $\Omega(n \log (m/n))$ communication lower bound.
For identical additive valuations, we additionally obtain a $1/2$-MMS guarantee within the same query complexity.
For non-identical additive valuations, exploiting our EF1+$1/2$-MMS algorithm, we accelerate the existing matching-based PROP1+$1/2$-MMS framework, improving the query complexity from $O(n^4\log m)$ to $O(n^3\log m)$.
Finally, we study the structure of such certificates through $k$-witness EF1, a hierarchy between EF1 and EFX.
\end{abstract}

\section{Introduction}

Fair division of indivisible goods studies how to allocate a set $M$ of $m$ indivisible items among $n$ agents so as to satisfy desirable fairness guarantees such as envy-freeness and proportionality \cite{Foley67,Steinhaus48,varian1973equity}.
When items are indivisible, exact envy-freeness (EF) or proportionality (PROP) may not be satisfied; consequently, prominent relaxations have been proposed, including 
\emph{envy-freeness up to one good} (EF1)~\cite{Bud2011,LMMS2004}, which requires that each agent's envy toward another agent can be eliminated by removing \emph{some} good from the envied bundle;
\emph{proportionality up to one good} (PROP1)~\cite{CFS2017}, which requires that each agent attain her proportional share after receiving some good outside her bundle; and the \emph{maximin share} (MMS)~\cite{Bud2011}, which is the value that an agent could receive by partitioning the goods into $n$ bundles and receiving the worst one (an $\alpha$-MMS allocation guarantees each agent an $\alpha$ fraction of her MMS)~\cite{LMMS2004,CFS2017,kurokawa2018fair}.
For additive valuations, EF1 and PROP1 allocations always exist~\cite{Bud2011,LMMS2004}, while exact MMS allocations may not exist and approximation guarantees are studied instead.
These existence results, however, presuppose that the valuations are fully known.
In this paper, we ask how efficiently such allocations can be found when the algorithm can access preferences only through ordinal comparisons.

\paragraph{The comparison-based query model.}
We adopt the comparison-based query model introduced by Bu, Li, Liu, Song, and Tao~\cite{BLLST2024}.
Each agent $i$ has a private cardinal valuation $u_i$ over $M$, and an algorithm may only ask agent $i$ to compare two bundles $X$ and $Y$; the answer is a single bit, consistent with $u_i$ but with arbitrary tie-breaking when $u_i(X)=u_i(Y)$.
This model captures settings in which agents can rank items but cannot (or should not) report cardinal values.
This is a weak but standard elicitation model for fair division, and it is strictly weaker than the value-query model, since one comparison can be simulated by two value queries.
The model has a notable consequence: because of adversarial tie-breaking of equality, even verifying that a given allocation is EF1 can be information-theoretically impossible (see \Cref{prop:ef1-undecidable-ties}).
An algorithm therefore cannot proceed by guess-and-check.
This suggests constructing, along with the allocation, an explicit \emph{witness certificate} of fairness, that is, a designated good in each bundle whose removal is guaranteed to eliminate envy.
This certificate viewpoint is the common thread of this paper.
Our main algorithm maintains and repairs explicit witnesses through a recursive scaling process, and the same perspective yields further applications.

\paragraph{Communication and description complexity.}
Feige~\cite{Fei25} studies the communication complexity of fair division with arbitrary queries about valuations, and introduces the notion of \emph{description complexity}: the number of bits needed to write down an acceptable allocation.
The description complexity of EF1 allocations is low, because there always exists a contiguous EF1 allocation of the goods~\cite{BCFIMPVZ22,Iga23}, which can be described by the $n-1$ cut points together with a permutation assigning the resulting intervals to the agents.
Feige~\cite{Fei25} shows that, for identical valuations, communication complexity coincides with description complexity, and thus both are $O(n\log (m/n))$.
For general additive valuations, the description complexity is only $O(n\log m)$, while the best known communication upper bound remains $O(m\log m)$.
On the hardness side, Feige proves that finding a PROP1 allocation requires $\Omega(n\log (m/n))$ bits of communication, even for randomized protocols and identical binary valuations; the same lower bound holds for computing an exact MMS allocation~\cite{Fei25}.
Since EF1 implies PROP1 for additive valuations and each comparison query uses one bit, this lower bound applies to EF1 in the comparison-query model.
We remark that the lower bound holds only for PROP1 (hence EF1) and \emph{exact} MMS.
For $\alpha$-MMS with $\alpha < 1$, the best known lower bound is $\Omega(\log(m/n))$~\cite{BLLST2024}.

\paragraph{From description to discovery.}
For identical valuations, the communication complexity of $O(n \log (m/n))$ is achieved by a trivial protocol: one agent computes a contiguous EF1 allocation by the algorithm in~\cite{BCFIMPVZ22}, and announces it using $O(n\log(m/n))$ bits.
However, this protocol cannot be carried to the comparison-query model, because the algorithm never learns any cardinal value.
A similar situation appears in comparison-based sorting, where the sorted order can be announced with $\lceil\log_2 n!\rceil$ bits by anyone who knows it, yet the canonical question is whether pairwise comparisons can achieve it.
Our comparison-query model asks whether an algorithm can \emph{discover} such an allocation through the weaker but standard form of preference elicitation.
There is one more difficulty.
As noted above, EF1 cannot be verified from comparison queries, and the algorithm must also find witnesses that certify EF1.
This leads to the following guiding question:
\begin{quote}
    \emph{Can an algorithm find certified fair allocations using only comparison queries at a cost close to their information content?
    In particular, for EF1 under identical valuations, can one achieve query complexity close to the $\Theta(n\log(m/n))$ benchmark?}
\end{quote}
The identical-valuation case is especially clean because description and communication complexity coincide.

\paragraph{Known algorithmic bounds.}
Bu, Li, Liu, Song, and Tao~\cite{BLLST2024} initiated the algorithmic study of the comparison-query model. 
Motivated by applications where the number of goods $m$ is extremely large, they provide $O(n^2\log m)$-query algorithms for EF1 and PROP1 under identical additive valuations, and $O(n^4\log m)$-query algorithms for PROP1 and PROP1+$1/2$-MMS under non-identical additive valuations.
They also prove an $\Omega(\log(m/n))$ lower bound that holds for EF1, PROP1, and $\alpha$-MMS for every $\alpha > 0$, already when the number $n$ of agents is a constant.
They further showed that an EF1 allocation can be found with $O(\log m)$ comparison queries for $n\le 3$ even with general additive valuations.

These bounds leave two different gaps.
For identical valuations, while logarithmic dependence on $m$ is already known, the best comparison-query upper bound has an extra factor of $n$ over the $\Theta(n\log(m/n))$ description/communication bound.
Thus, in this case, improving the dependence on $n$ asks whether comparison queries can discover and certify an EF1 allocation with essentially the same number of bits as are needed just to describe one.

For non-identical valuations, the situation is more challenging.
Whether EF1 can be achieved with $\mathrm{poly}(n,\log m)$ comparison queries is open already for $n=4$.
In fact, this difficulty is not specific to comparison queries.
Even in the stronger communication model of Feige~\cite{Fei25}, the best upper bound for EF1 under general additive valuations is $O(m\log m)$ by round-robin, while the description complexity is only $O(n\log m)$.
Thus, closing the non-identical EF1 gap appears to require new ideas even before imposing the restriction to comparison queries.

\paragraph{Our main result.}
For identical valuations, we close the gap to the description complexity up to a single $O(\log n)$ factor.
\Cref{tab:main-results} summarizes our results.
Our main technical contribution is a certificate-preserving scaling framework: it recursively contracts pairs of goods into meta-goods, solves the smaller instance, and expands the solution one level at a time while repairing an explicit witness for envy elimination.
For arbitrary identical \emph{monotone} valuations, the framework computes a certified EF1 allocation using $O(n\log n\log(m/n))$ comparison queries (\Cref{thm:sendmaxtomin}), which matches the $\Omega(n\log(m/n))$ lower bound up to an $O(\log n)$ factor.
Moreover, the algorithm outputs an explicit witness for each bundle, which certifies EF1 even in the comparison-query model.
We call such an allocation \emph{1-witness EF1}\footnote{This solution concept was introduced by Conitzer, Freeman, Shah, and Wortman Vaughan~\cite{CFSV19} under the name \emph{strong EF1}, and the existence is already shown via the Nash welfare argument.
Our contribution is not existence but query-efficient construction. We adopt the witness terminology to emphasize the role of the designated goods and the extension to $k$ witnesses.}.
Since each comparison query uses one bit, our algorithm is also a communication protocol.

For identical \emph{additive} valuations, we strengthen the guarantee at no extra cost.
Within the same $O(n\log n\log(m/n))$ query complexity, we can compute an allocation that is simultaneously EF1 and $1/2$-TPS (\Cref{thm:ef1-halfmms}).
Here, TPS denotes the \emph{truncated proportional share}~\cite{BEF22}, lying between the MMS and PROP values.
In particular, the allocation is $1/2$-MMS.

\paragraph{Applications of the certificate machinery.}
The witness certificates produced by our framework can be used as algorithmic primitives for non-identical instances and as structural objects.

First, for \emph{non-identical} additive valuations, we accelerate the matching-based framework of~\cite{BLLST2024} for PROP1+$1/2$-MMS by one polynomial factor, from $O(n^4\log m)$ to $O(n^3\log m)$ queries.
Moreover, our guarantee is strengthened to PROP1 plus $1/2$-TPS (\Cref{thm:prop1-half-mms-non-identical}).
The speedup comes from an efficient construction of an underlying bipartite graph by using witness certificates for each agent obtained from our EF1 algorithm.
Note that our result achieves almost the same guarantee as Feige's Aprop protocol~\cite{Fei25} (see \Cref{sec:further-related-work}) with only comparison queries.

Second, we make partial progress on EF1 allocations under non-identical valuations by using witness detection.
If, for some unknown integer $P$, there exists a matching of each agent $i$ to a good $g$ with $u_i(g)>u_i(M)/(nP)$, then we can compute a 1-witness EF1 allocation using $O(\mathrm{poly}(n,P)\log m)$ queries (\Cref{prop:non-identical-EF1-heavy-anchors}).
A similar argument yields an $O(\mathrm{poly}(n,K)\log m)$-query algorithm for binary valuations with hidden support size at most $K$ (\Cref{prop:non-identical-EF1-binary}).

Third, we study the structure of witness certificates through a hierarchy called \emph{$k$-witness EF1}.
This notion requires that the last $k$ goods of every bundle must all serve as EF1 witnesses.
While our algorithms above construct 1-witness certificates, the hierarchy asks how much fairness is guaranteed when more goods in each bundle are certified as envy-eliminating witnesses. 

Note that $m$-witness EF1 coincides with envy-freeness up to \emph{any} good (EFX)~\cite{GMT2014,PR2020}, which requires that envy is eliminated by the removal of \emph{any} good from the envied bundle.
Thus, the hierarchy interpolates between 1-witness EF1 and EFX.
We show that, under general additive valuations, every 2-witness EF1 allocation is $1/2$-TPS (\Cref{thm:2-witness-implies-half-mms}), and the TPS guarantee cannot be improved for $k$-witness EF1 with $k\ge 2$.
We also show that every 3-witness EF1 allocation is $4/7$-MMS (\Cref{thm:3-witness-implies-47-mms}), and the factor cannot exceed $10/17$ (\Cref{ex:3-witness-example}).
We note that our lower bound on the approximation ratio matches the current best guarantee for EFX~\cite{ABM2018}, and explicit witnesses form a structural relaxation of EFX that already captures its known MMS guarantee at $k=3$.
Thus, the witness viewpoint is also a structural relaxation of EFX.

\begin{table}[ht]
    \centering
    \small
    \caption{
        Algorithms for fair division of indivisible goods in the comparison-based query model. The lower bound $\Omega(n\log (m/n))$ of~\cite{Fei25} applies to all rows, since EF1 implies PROP1 under additive valuations.
        We note that $1/2$-TPS implies $1/2$-MMS.
        For $\alpha$-MMS with $\alpha<1$, only the weaker lower bound $\Omega(\log(m/n))$ of~\cite{BLLST2024} is known.
        We also obtain parameterized EF1 algorithms for restricted non-identical settings (\Cref{prop:non-identical-EF1-heavy-anchors,prop:non-identical-EF1-binary}).
    }
    \label{tab:main-results}
    \begin{tabular}{lllll}
        \toprule
        Algorithm & Identical? & Additive? & Fairness & \makecell[l]{Query Complexity} \\
        \midrule
        \makecell[l]{\cite{BLLST2024}} & Identical & Additive & EF1 & $O(n^2 \log m)$ \\
        \hline
        \makecell[l]{\textbf{Ours} \\ (\Cref{thm:sendmaxtomin})} & Identical & \textbf{Monotone} & EF1 & $\boldsymbol{O(n \log n \log (m/n))}$ \\
        \hline
        \makecell[l]{\textbf{Ours} \\ (\Cref{thm:ef1-halfmms})} 
        & Identical & Additive & \makecell[l]{EF1 \\ + \textbf{$1/2$-TPS}} & $\boldsymbol{O(n \log n \log (m/n))}$ \\
        \hline
        \hline
        \makecell[l]{\cite{BLLST2024}} & Non-identical & Additive & \makecell[l]{PROP1 \\ + $1/2$-MMS} & $O(n^4 \log m)$ \\
        \hline
        \makecell[l]{\textbf{Ours} \\ (\Cref{thm:prop1-half-mms-non-identical})} 
        & Non-identical & Additive & \makecell[l]{PROP1 \\ + \textbf{$1/2$-TPS}} & $\boldsymbol{O(n^3 \log m)}$ \\
        \bottomrule
    \end{tabular}
\end{table}

\paragraph{Optimality and Lower bounds.}
For identical valuations, the communication lower bound of $\Omega(n\log(m/n))$ also applies to comparison queries, while our upper bound is $O(n\log n\log(m/n))$.
In our framework, the remaining $O(\log n)$ factor comes from a sorting barrier in witness repair tasks within each scaling level (Appendix~\ref{app:witness-repair-lb}).
This suggests that removing the remaining $O(\log n)$ factor might need a different approach.
Under the non-identical setting, our $O(n^3\log m)$ bound for PROP1 should be contrasted with the same $\Omega(n\log(m/n))$ lower bound, and no $\mathrm{poly}(n,\log m)$ upper bound for EF1 with $n\ge 4$ agents is known.
Even in Feige's stronger communication model, the best known EF1 protocol for general additive valuations uses $O(m\log m)$ bits via round-robin, while the existence of $O(n\log m)$-bit descriptions is implied from contiguous EF1 allocations~\cite{Iga23}.

\subsection{Technical Overview}
Recall two features of the model.
First, as the target query bound is $O(\mathrm{poly}(n, \log m))$, almost all goods never appear in any query as individuals.
This rules out simulating standard algorithms such as round-robin or envy-cycle elimination, which scan every good.
Second, because EF1 cannot be verified from comparisons, we need to construct an explicit witness to eliminate envy throughout; the technical work lies in repairing witnesses cheaply.

Our main technical tool is a ``scaling'' framework that repeatedly pairs goods and contracts each pair into a single meta-good, shrinking the instance size geometrically.
At each level, we solve an intermediate problem of finding a 1-witness EF1 allocation using a transfer procedure, \textsf{SendMaxToMin}, which moves a designated ``last'' good from an unmarked bundle with the current maximum value to the minimum bundle whenever the comparison queries indicate that the transfer does not decrease the minimum of the two.
The obstacle is that expanding a witness meta-good into two (meta-)goods may destroy the witness certificate; the last good of a bundle may no longer certify EF1 after expansion.
The key insight is that, after expanding one level of scaling, removing the last two goods from any bundle never yields a bundle higher than the current minimum.
This condition allows the witnesses to be repaired locally.
A refined transfer procedure, $\SendMaxToMin'$, marks both the sender and receiver bundles.
Each iteration newly marks at least one bundle, and hence there are at most $n$ iterations per level.
Each transfer uses $O(\log n)$ comparison queries to reinsert the modified bundles into the sorted order of bundles.
Thus, the repair phase finishes in $O(n \log n)$ comparison queries per scaling level, yielding the overall $O(n \log n \log (m/n))$ bound.
We remark that the procedure $\SendMaxToMin$ can also find an EF1 allocation, but it is not efficient on its own.

A 1-witness EF1 allocation can be used to identify ``high-value'' goods (\Cref{lem:large-goods-in-witnesses}).
In a 1-witness EF1 allocation with $k$ bundles for an identical additive valuation $u$, any singleton good whose value exceeds the minimum bundle value must itself be a designated witness for EF1.
Since the minimum bundle has value at most $u(M)/k$, the witness set contains all goods of value higher than $u(M)/k$.

We utilize this fact to obtain the guarantee of EF1 plus $1/2$-TPS for identical additive valuations.
Appropriately dealing with high-value goods is necessary for $1/2$-TPS.
We first identify a set $W_0$ containing all goods with values higher than $u(M)/(2n+1)$.
We then allocate only these witness goods using the Egal-Sequential algorithm~\cite{AR2020} (whose output is EFX and hence $1$-witness EF1), append the remaining goods, and finally re-run the scaling-based EF1 routine while ensuring that witness goods stay in their assigned bundles.
This yields an EF1 allocation where every good with value higher than $u(M)/n$ forms a singleton bundle, which guarantees the $1/2$-TPS property.

For the non-identical additive case, we adopt the high-level strategy introduced in \cite{BLLST2024}: construct a bipartite graph between agents and candidate bundles which we call the \emph{want-this graph}, and select a non-empty matching via a \emph{Hall matching} argument.
Our contribution is to speed up the construction of the want-this graph.
By reusing our 1-witness EF1 algorithm, we preprocess each agent $i$ to obtain a two-field certificate $(T_i,g_i)$.
The threshold bundle $T_i$ quickly certifies that a candidate bundle is non-PROP (the PROP condition is not satisfied) for agent $i$; the witness good $g_i$ certifies that every candidate bundle not containing $g_i$ and not below the threshold is PROP1, with $g_i$ serving as the witness.
This allows most edge decisions in the want-this graph to be made using only $O(1)$ comparisons.
The only exceptional case is when the candidate bundle contains the single witness good $g_i$.
This can happen for at most one bundle per agent in each iteration, and in that case, we invoke a subroutine for deciding PROP1 versus non-PROP given in \cite{BLLST2024}.

To guarantee simultaneously PROP1 and $1/2$-TPS for non-identical additive valuations, we modify the above algorithm by replacing the $1$-witness EF1 allocation used for threshold computation with a $1$-witness EF1 plus $1/2$-TPS allocation, and prioritizing the assignment of singletons that satisfy the PROP1 condition.
This approach strengthens the $1/2$-MMS guarantee of~\cite{BLLST2024} to $1/2$-TPS.

\Cref{lem:large-goods-in-witnesses} also gives partial EF1 results for non-identical additive valuations.
By witness detection, we can extract high-value goods for each agent.
We guess a scale $q$ by doubling trick, and we run separately $\ScalingSendMaxToMin$ for $q$ bundles with respect to each agent's valuation.
Once $q$ becomes large enough, the witness sets contain all high-value goods for every agent.
Thus, if such goods exist, we can assign such goods to agents so that the goods can be designated witnesses and the remaining goods have bounded values.
Note that those goods alone do not guarantee EF1, and also the remaining goods must be compressed.
By using the matched high-value goods as thresholds, we can compress the remaining goods into a small number of meta-goods.
Then the round-robin algorithm for the high-value goods and meta-goods yields a 1-witness EF1 allocation.

The final part of the paper asks what maintaining two or three witnesses per bundle would imply.
Our proof of the $4/7$-MMS guarantee for 3-witness EF1 (\Cref{thm:3-witness-implies-47-mms}) shares its essential idea with the argument for EFX in~\cite{ABM2018}.
Both arguments proceed by repeatedly reducing a given allocation, and they differ in the operation applied to bundles of size two.

\subsection{Further related work}\label{sec:further-related-work}
A classical way to model incomplete preferences is the value-query model, where the algorithm may ask for the numerical value $u_i(S)$ of a bundle $S$~\cite{LMMS2004,OPS2021}.
Query complexity varies dramatically across fairness notions in this model.
Oh, Procaccia, and Suksompong~\cite{OPS2021} showed that for two agents with monotone valuations an EF1 allocation can be found with $O(\log m)$ value queries, with polylogarithmic bounds for three agents, whereas EFX requires exponentially many value queries already for two agents in general classes~\cite{PR2020}, and linearly many for two agents with identical additive valuations~\cite{OPS2021}.
Communication complexity has also been studied for discrete fair division~\cite{PR2020b} and cake cutting~\cite{BN2019}.

Our PROP1+$1/2$-TPS result can be compared with Feige's Aprop protocol~\cite{Fei25}, which achieves PROP1 together with $\frac{n}{2n-1}$-TPS using $O(n\log n\log m)$ bits of communication (and expected $O(n\log m)$ with randomization).
However, these protocols use stronger queries that ask agents to compare valuations rescaled by specific factors, which comparison queries cannot do.
 
Many existence proofs for EF1 rest on sequential picking rules such as round-robin, but correctly implementing such rules from ordinal information is itself costly.
It is shown in~\cite{LMSS24} that reproducing the round-robin outcome from item comparisons requires $\Omega(nm+m\log m)$ queries even with randomization.
This contrast further illustrates why near-logarithmic dependence on $m$ requires algorithms designed natively for the comparison-query model, rather than simulations of classical fair-division algorithms.


\section{Preliminaries}

In this section, we introduce the fair allocation problem under the comparison-based query model.

We are given a set of $n$ agents labeled $1, 2, \dots, n$, and a set $M = \{g_1, g_2, \dots, g_m\}$ of $m$ goods. We denote the set of agents by $[n] \coloneqq \{1, 2, \dots, n\}$.
A subset of $M$ is called a \emph{bundle}.
    
An \emph{allocation} is a partition of the goods into $n$ bundles, denoted by $\A = (A_1, A_2, \dots, A_n)$, where a bundle $A_i$ is to be allocated to agent $i$ for each $i \in [n]$.
We note that $A_i \cap A_j = \emptyset$ for any distinct $i, j\in [n]$ and $\bigcup_{i\in [n]} A_i = M$.
A subpartition of $M$, in which some goods remain unallocated, is called a \emph{partial allocation}.

Each agent $i \in [n]$ has a \emph{valuation function} $u_i\colon 2^M \to \mathbb{R}_{\ge 0}$ that represents how much the agent values a given bundle.
For every good $g \in M$ and a valuation function $u$, we abbreviate $u(\{g\})$ as $u(g)$. 
A valuation function $u$ is said to be \emph{normalized} if $u(\emptyset) = 0$, and \emph{monotone} if for all $X, Y \subseteq M$, if $X \subseteq Y$, then $u(X) \le u(Y)$.
In this paper, we assume that each $u_i$ is normalized and monotone.

A valuation function $u_i$ is said to be \emph{additive} if $u_i(X) = \sum_{g \in X} u_i(g)$ for all $X \subseteq M$.
Valuation functions $u_1, \dots, u_n$ are said to be \emph{identical} if there exists a valuation function $u$ such that $u_i = u$ for all $i \in [n]$.
We sometimes assume that valuation functions are additive or identical.

The goal of our problem is to find an allocation $\A = (A_1, A_2, \dots, A_n)$ that satisfies fairness criteria, which will be defined later.

For a bundle $X \subseteq M$ and a good $x \in M$, we write $X + x$ for $X \cup \{x\}$ and $X - x$ for $X \setminus \{x\}$.

For the sake of exposition, we treat bundles as ordered lists.
For a nonempty bundle $X$, we denote the last good of $X$ by $\last(X)$, and the $k$-th good from the end by $\last_k(X)$.


\subsection{Comparison-Based Query Model}

We adopt the comparison-query model introduced by \cite{BLLST2024} for the fair allocation of indivisible goods.
In this model, valuation functions are not directly provided as input and can only be accessed via \emph{comparison-based queries}.

\begin{definition}[Comparison-Based Query]
    The algorithm can access agents' valuations via the following query:
    \begin{quote}
        $\Comp_i(X, Y)$: Given two bundles $X, Y \subseteq M$ (possibly overlapping), query agent $i$ to determine which bundle they prefer.
        The query returns $\true$ if $u_i(X) < u_i(Y)$, returns $\false$ if $u_i(X) > u_i(Y)$, and may return either $\true$ or $\false$ if $u_i(X) = u_i(Y)$.
    \end{quote}
    \noindent When the valuations are identical, we abbreviate $\Comp_i$ as $\Comp$.
\end{definition}
In this model, we only know whether $u_i(X)\le u_i(Y)$ or $u_i(X)\ge u_i(Y)$. Thus, the algorithm cannot certify $u_i(X)= u_i(Y)$, nor can it distinguish a preference with ties from a strict one.
The answers of queries can form a cycle of inequalities, such as $\Comp_i(X, Y)=\true$, $\Comp_i(Y, Z)=\true$ and $\Comp_i(Z, X)=\true$, but the cycle means $u_i(X)=u_i(Y)=u_i(Z)$.
Thus, we can correctly sort goods or bundles with comparison queries.

Throughout this paper, we assume that comparison queries dominate the computation time, and we therefore evaluate only the query complexity.
We manage bundles using explicit ordered lists, and do not focus on the running time.
For \Cref{alg:prop1-non-identical,alg:prop1-half-mms-non-identical}, the time complexity of bipartite matching may become the bottleneck.


\subsection{Fairness Criteria}

We now introduce several fairness criteria.
The most fundamental fairness criteria are envy-freeness (EF) and proportionality (PROP).
\begin{definition}[EF]
    An allocation $\A = (A_1, A_2, \dots, A_n)$ is said to be \emph{envy-free} (EF) if 
    $u_i(A_i) \ge u_i(A_j)$ for all agents $i, j \in [n]$.
\end{definition}

\begin{definition}[PROP]
    An allocation $\A = (A_1, A_2, \dots, A_n)$ is said to be \emph{proportional} (PROP) if 
    $u_i(A_i) \ge u_i(M) / n$ for all agents $i \in [n]$.
\end{definition}

It is well known that allocations satisfying EF or PROP may not always exist.
When there are two agents and one good with a positive value, allocating the good to either agent is not EF nor PROP.
To address this issue, relaxed fairness notions that allow for the margin of one good have been proposed.

\begin{definition}[EF1~\cite{LMMS2004,Bud2011}]
    An allocation $\A = (A_1, A_2, \dots, A_n)$ is said to be \emph{envy-free up to one good} (EF1) if for all agents $i, j \in [n]$, either $u_i(A_i) \ge u_i(A_j)$ or there exists a good $g \in A_j$ such that $u_i(A_i) \ge u_i(A_j - g)$.
\end{definition}

\begin{definition}[PROP1~\cite{CFS2017}]
    An allocation $\A = (A_1, A_2, \dots, A_n)$ is said to be \emph{proportional up to one good} (PROP1) if for all agents $i \in [n]$, either $u_i(A_i) \ge u_i(M) / n$ or there exists a good $g \in M \setminus A_i$ such that $u_i(A_i + g) \ge u_i(M) / n$.
\end{definition}

Under additive valuations, any EF1 allocation $\A$ is PROP1.
Indeed, if $u_i(A_i)<u_i(M)/n$, then some bundle $A_j$ satisfies $u_i(A_j)>u_i(M)/n$.
Since $\A$ is EF1, there exists $g\in A_j$ such that $u_i(A_i)\ge u_i(A_j-g)$, and hence, $u_i(A_i+g)\ge u_i(A_j)>u_i(M)/n$ holds.

It is known that an EF1 (and hence PROP1) allocation always exists.
In particular, the envy-cycle elimination algorithm~\cite{LMMS2004} finds an EF1 allocation using $O(mn^3)$ comparison queries.

We provide the following observation on PROP1 allocations, which will be used later.
\begin{lemma}
\label{lem:prop1-singleton}
    Assume that valuations are additive.
    Let $\A = (A_1, A_2, \dots, A_n)$ be an allocation satisfying $|A_n| = 1$.
    Fix an agent $i \in [n-1]$.
    If $A_i$ satisfies either $u_i(A_i) \ge u_i(M \setminus A_n) / (n - 1)$ or there exists a good $g \in M \setminus (A_n\cup A_i)$ such that $u_i(A_i + g) \ge u_i(M \setminus A_n) / (n - 1)$, then $A_i$ also satisfies either $u_i(A_i) \ge u_i(M) / n$ or there exists a good $g' \in M \setminus A_i$ such that $u_i(A_i + g') \ge u_i(M) / n$.
\end{lemma}
\begin{proof}
    If $u_i(A_n) \le u_i(M) / n$, then $u_i(M \setminus A_n) / (n - 1) \ge u_i(M) / n$, which implies the claim trivially.
    Otherwise, choosing $g' \in A_n$ yields $u_i(A_i + g') \ge u_i(M) / n$, and thus the claim holds.
\end{proof}

EFX is a stronger fairness notion derived from EF1. 
When valuations are identical, an EFX allocation always exists~\cite{PR2020}.
For additive valuations, the existence of EFX allocations is a major open problem~\cite{Ama2023}; recently, counterexamples were found for general monotone valuations~\cite{AMMSW26} and, with a compact construction, for submodular valuations~\cite{MS26}.

\begin{definition}[EFX~\cite{GMT2014,PR2020}]
    An allocation $\A = (A_1, A_2, \dots, A_n)$ is said to be \emph{envy-free up to any good} (EFX) if for all agents $i, j \in [n]$ and for all goods $g \in A_j$, we have $u_i(A_i) \ge u_i(A_j - g)$.
\end{definition}

For additive valuations, another natural fairness criterion is the maximin share.
\begin{definition}[Maximin Share Fairness~\cite{Bud2011}]
    For an agent $i \in [n]$, the \emph{maximin share} of agent $i$ is defined as
    \begin{align*}
        \mu_i^n(M) \coloneqq \max_{\B \in \Pi_n(M)} \min_{B \in \B} u_i(B),
    \end{align*}
    where $\Pi_n(M)$ is the set of all possible partitions of $M$ into $n$ bundles.
    Clearly, $\mu_i^n$ satisfies $\mu_i^n(M) \le u_i(M) / n$.
    When valuations are identical, we omit the subscript $i$ and write $\mu^n(M)$.
    
    An allocation $\A = (A_1, A_2, \dots, A_n)$ is said to be \emph{maximin share fair} (MMS) if 
    $u_i(A_i) \ge \mu_i^n(M)$ for all agents $i \in [n]$.
    For $\alpha \in (0, 1]$, an allocation $\A = (A_1, A_2, \dots, A_n)$ is said to satisfy $\alpha$-MMS if 
    $u_i(A_i) \ge \alpha \cdot \mu_i^n(M)$ for all agents $i \in [n]$.
\end{definition}

An exact MMS allocation may not exist.
For existence of $\alpha$-MMS allocations, the best possible approximation ratio is known to lie between $7/9$~\cite{HZ2025} and $39/40$~\cite{FST21}.

For additive valuations, we also use the following refinement of the proportional share.
\begin{definition}[Truncated Proportional Share~\cite{BEF22}]
    For an agent $i \in [n]$ and a value $\tau \ge 0$, let 
    $L_i(\tau) \coloneqq \{g \in M \mid u_i(g) > \tau\}$.
    The \emph{truncated proportional share} of agent $i$ is defined as
    \begin{align*}
        \TPS_i^n(M) \coloneqq
        \max \left\{ \tau \ge 0 \ \middle|\
        |L_i(\tau)| < n
        \text{ and }
        u_i(M \setminus L_i(\tau)) \ge (n - |L_i(\tau)|)\tau
        \right\}.
    \end{align*}
    When valuations are identical, we omit the subscript $i$ and write $\TPS^n(M)$.
    For $\alpha \in (0,1]$, an allocation $\A = (A_1,A_2,\dots,A_n)$ is said to satisfy $\alpha$-TPS if
    $u_i(A_i) \ge \alpha \cdot \TPS_i^n(M)$ for all agents $i \in [n]$.
\end{definition}

It is known that $\mu_i^n(M) \le \TPS_i^n(M) \le u_i(M)/n$ for additive valuations~\cite{BEF22}.
We will use the following property of the maximin share in later proofs.
The truncated proportional share satisfies the analogous property.
While \Cref{lem:TPS-singleton} is stated in the full version of~\cite{BEF22} (within the proof of Proposition~20, without an explicit proof), we include a proof for completeness.

\begin{lemma}[\cite{BL2016,AMNS2017}]
\label{lem:MMS-singleton}
    For any good $g \in M$ and $n \ge 2$, we have $\mu_i^{n-1}(M - g) \ge \mu_i^{n}(M)$ for any agent $i\in [n]$.
\end{lemma}

\begin{lemma}[\cite{BEF22}]\label{lem:TPS-singleton}
    For any good $g \in M$ and $n \ge 2$, we have $\TPS_i^{n-1}(M - g) \ge \TPS_i^n(M)$ for any agent $i\in [n]$.
\end{lemma}
\begin{proof}
    Let $\tau = \TPS_i^n(M)$ and let $L = L_i(\tau)$.
    We show that $\TPS_i^{n-1}(M-g) \ge \tau$.
    If $g \in L$, then the goods with value greater than $\tau$ in $M-g$ are exactly $L-g$.
    Hence, by the definition of TPS, 
    \begin{align*}
        u_i((M-g)\setminus (L-g))
        =
        u_i(M\setminus L)
        \ge (n-|L|)\tau
        =
        ((n-1)-(|L|-1))\tau .
    \end{align*}
    If $g \notin L$, then the goods with value greater than $\tau$ in $M-g$ are exactly $L$.
    If $|L| = n-1$, then $M-g$ contains the $n-1$ goods in $L$, each having value greater than $\tau$.
    Hence $\mu_i^{n-1}(M-g) > \tau$, and since $\TPS_i^{n-1}(M-g) \ge \mu_i^{n-1}(M-g)$, we obtain $\TPS_i^{n-1}(M-g) > \tau$.
    It remains to consider the case $|L| < n-1$.
    Since $u_i(g) \le \tau$, we have
    \begin{align*}
        u_i((M-g)\setminus L)
        =
        u_i(M\setminus L)-u_i(g)
        \ge (n-|L|)\tau-\tau
        =
        ((n-1)-|L|)\tau .
    \end{align*}
    Thus, in either case, $\TPS_i^{n-1}(M-g) \ge \tau=\TPS_i^n(M)$.
\end{proof}


\section{EF1 Allocation for Identical and Monotone Valuations}

In this section, we present algorithms to compute EF1 allocations when all agents share the same valuation function.
Bu, Li, Liu, Song, and Tao~\cite{BLLST2024} showed that for identical additive valuations, EF1 can be found with $O(n^2\log m)$ comparison queries.

Our main result is an $O(n\log n\log (m/n))$-query algorithm that finds an EF1 allocation under merely monotone identical valuations (Theorem~\ref{thm:sendmaxtomin}), improving the known bounds and removing the additivity assumption.
The key ingredients are a scaling framework that merges goods into meta-goods across levels, and a transfer routine called \textsf{SendMaxToMin} that rebalances bundles efficiently using only comparisons.

We first observe that, because equality is not observable and ties may be broken arbitrarily, we cannot determine deterministically even whether a given allocation is EF1 or not.
\begin{proposition}
\label{prop:ef1-undecidable-ties}
    In the comparison-query model with arbitrary tie-breaking on equal values, no deterministic algorithm can decide whether a given allocation is EF1, even for two agents with identical additive valuations.
\end{proposition}

The proof is found in Appendix~\ref{app:proof-section-3}.
This suggests constructing an allocation together with explicit witness goods, rather than attempting to verify EF1.
This motivates us to focus on witness-certified EF1 allocations.
\begin{definition}
An allocation $\A = (A_1, A_2, \dots, A_n)$ is said to satisfy \emph{1-witness EF1} if for all agents $i, j \in [n]$, either $u_i(A_i) \ge u_i(A_j)$ or $u_i(A_i)\ge u_i(A_j - \last(A_j))$.
\end{definition}

Here we remove the last good of $A_j$ only to fix a canonical witness for later algorithms; the specific position of the witness good is not important.
Clearly, 1-witness EF1 implies EF1, since it specifies a particular good whose removal eliminates envy whenever envy exists.
We remark that this notion is also studied under the name strong EF1~\cite{CFSV19}.

Throughout this section, we assume that the agents' valuations are identical.

\subsection{Warm-up: {\SendMaxToMin}}

We first present a procedure $\SendMaxToMin$ that computes a 1-witness EF1 allocation under identical valuations in \Cref{alg:sendmaxtomin}.
This algorithm maintains bundles in nondecreasing order of value.
In each iteration, it selects the unmarked bundle $A_i$ with maximum value.
If the comparison query oracle returns that $u(A_1) \le u(A_i - \last(A_i))$, then $\A$ may not be EF1, and hence the algorithm transfers $A_i$'s last good to $A_1$, marks the bundle $A_1$, and sorts the  bundle order by binary insertion.
Otherwise, the bundle $A_i$ is marked.
A bundle becomes marked once its last good serves as a valid witness against the current minimum.

\begin{algorithm}[ht]
\caption{Computing a 1-witness EF1 allocation naively under identical valuations}
\label{alg:sendmaxtomin}
\Function {$\SendMaxToMin(u, n, \A = (A_1, A_2, \dots, A_n))$}{
    \Precondition{$u(A_1) \le u(A_2) \le \dots \le u(A_n)$}
    \For {$i \in [n]$}{
        \Comment{Each bundle maintains a state of either marked or unmarked}
        Set $A_i$ as unmarked\;
    }
    \While {$\A$ has unmarked bundles} {
        \Comment{Invariant: $u(A_1) \le u(A_2) \le \dots \le u(A_n)$ at the start/end of each iteration}
        Let $i$ be the maximum index in $[n]$ such that $A_i$ is unmarked\label{line:sendmaxtomin_unmarked}\;
        \If {$|A_i| \ge 1$ \emph{and} $\Comp(A_1, A_i - \last(A_i))$ \label{line:sendmaxtomin_compare}} {
            Remove $\last(A_i)$ from $A_i$ and append it to the end of $A_1$\label{line:sendmaxtomin_transfer}\;
            Set $A_1$ as marked\label{line:sendmaxtomin_mark}\;
            Reorder $\A$ to maintain $u(A_1) \le u(A_2) \le \dots \le u(A_n)$\label{line:sendmaxtomin_reorder}\;
            \Comment{Use binary insertion with $O(\log n)$ comparison queries}
        }
        \Else {
            Set $A_i$ as marked\label{line:sendmaxtomin_mark2}\;
        }
    }
    \Return $\A = (A_1, A_2, \dots, A_n)$\;
}
\end{algorithm}

We now show the correctness of \Cref{alg:sendmaxtomin}, and we establish several additional properties that will be used later in the scaling framework.

\begin{lemma}
\label{lem:non-decreasing-min}
    During the execution of \Cref{alg:sendmaxtomin}, $\min_{A \in \A} u(A)$ is non-decreasing.
\end{lemma}
\begin{proof}
    The only modification to bundle contents occurs at line~\ref{line:sendmaxtomin_transfer}, where a good is transferred.
    It suffices to show that $\min_{A \in \A} u(A)$ does not decrease at this step.
    By the precondition and the sorting at line~\ref{line:sendmaxtomin_reorder}, we have $u(A_1) \le u(A_2) \le \dots \le u(A_n)$ at line~\ref{line:sendmaxtomin_unmarked}.
    Before executing line~\ref{line:sendmaxtomin_transfer}, the if-condition ensures that $u(A_1) \le u(A_i - \last(A_i))$, and by monotonicity of $u$, we have $u(A_1) \le u(A_1 + \last(A_i))$.
    Since no other bundles are modified, $\min_{A \in \A} u(A)$ does not decrease at line~\ref{line:sendmaxtomin_transfer}.
\end{proof}

\begin{lemma}
\label{lem:sendmaxtomin-is-witness-EF1}
    The output of \Cref{alg:sendmaxtomin} is a 1-witness EF1 allocation.
\end{lemma}
\begin{proof}
    Let $\A = (A_1, A_2, \dots, A_n)$ be the output of \Cref{alg:sendmaxtomin}.
    This allocation satisfies $u(A_1) \le u(A_2) \le \dots \le u(A_n)$.

    We show that at the end of each iteration of the while loop, every  marked bundle $X \in \A$ satisfies either $u(X) \le u(A_1)$ or $u(X - \last(X)) \le u(A_1)$.

    \UnderlineCase{Base case:} In each iteration, when some bundle is first marked (either at line~\ref{line:sendmaxtomin_mark} or line~\ref{line:sendmaxtomin_mark2}), the condition of this lemma clearly holds for the newly marked bundle at the end of that iteration.

    \UnderlineCase{Inductive step:} Assume that the condition holds for each marked bundle $X$ at the beginning of an iteration. 
    We show that this holds at the end of this iteration.
    Let $a$ and $a'$ be the minimum values of the bundles at the beginning and end of the iteration, respectively.
    By \Cref{lem:non-decreasing-min}, we have $a \le a'$.
    
    By the inductive hypothesis, we have $u(X) \le a$ or $u(X - \last(X)) \le a$.
    If bundle $X$ is not modified during the iteration, then we still have $u(X) \le a'$ or $u(X - \last(X)) \le a'$, and hence the condition is maintained.
    If bundle $X$ is modified to $X'$ during the iteration, then since $X$ is marked, $X$ must be added a good.
    Then $u(X'-\last(X')) = u(X) \leq a \leq a'$ holds, and $X'$ is marked at line~\ref{line:sendmaxtomin_mark}.
    Therefore, the condition holds at the end of the iteration.
    
    Furthermore, upon termination of \Cref{alg:sendmaxtomin}, all bundles are marked. 
    Hence, the output of \Cref{alg:sendmaxtomin} is a 1-witness EF1 allocation.
\end{proof}

\begin{lemma}
\label{lem:sendmaxtomin-complexity}
    Let $m \coloneqq \sum_{A \in \A} |A|$. Then \Cref{alg:sendmaxtomin} uses $O(m \log n)$ comparison queries.
\end{lemma}
\begin{proof}
    Comparison queries are invoked at line~\ref{line:sendmaxtomin_compare} when $|A_i| \ge 1$ and at line~\ref{line:sendmaxtomin_reorder}.
    In each iteration where $A_i$ is nonempty, the total number of goods in marked bundles increases by at least one.
    Hence the number of such iterations is at most $m$.
    Line~\ref{line:sendmaxtomin_compare} uses one comparison query, and line~\ref{line:sendmaxtomin_reorder} uses $O(\log n)$ comparison queries, yielding $O(m \log n)$ comparison queries overall.
\end{proof}

\subsection{{\SendMaxToMin} with Scaling}

The procedure $\SendMaxToMin$ is conceptually simple, but a direct implementation may require $\Theta(m)$ successful transfers.
To reduce the overall number of comparison queries, we embed the transfer procedure into a scaling framework that shrinks the instance size geometrically.

Our scaling framework $\ScalingSendMaxToMin$ repeatedly pairs adjacent goods within each bundle and contracts each pair into a single meta-good, thereby producing a sequence of coarser instances of almost half the size.
This is equivalent to doubling the number of goods transferred at each step in $\SendMaxToMin$.
We first compute a 1-witness EF1 allocation on the coarsest instance, and then iteratively expand the allocation back to the original goods while preserving their order within each bundle.
The main benefit is that, after each expansion step, the allocation satisfies a stronger structural precondition: removing the last two goods from any bundle never yields a bundle higher than the current minimum.
This additional condition allows us to apply a refined version of $\SendMaxToMin$, which we call $\SendMaxToMin'$, to terminate with only $O(n\log n)$ comparisons per level, leading to the overall $O(n\log n\log (m/n))$ bound.
The procedure $\SendMaxToMin'$ differs from $\SendMaxToMin$ in that it marks a bundle sending a good as well as a bundle receiving a good.
We present the formal description of $\ScalingSendMaxToMin$ in~\Cref{alg:scaling-sendmaxtomin}.

We identify a bundle of meta-goods with the union of the original goods represented by its elements.
The value of a bundle of meta-goods and the comparisons between two such bundles are defined by the corresponding unions of original goods.
Thus, each comparison between such bundles uses one comparison query to the original oracle.
We use the same notation $u$ and $\Comp$ for these induced operations.

\begin{algorithm}[htb]
\caption{Computing a 1-witness EF1 allocation under identical valuations}
\label{alg:scaling-sendmaxtomin}
\Function {$\SendMaxToMin'(u, n, \A = (A_1, A_2, \dots, A_n))$}{
    \Precondition{$u(A_1) \le u(A_2) \le \dots \le u(A_n)$}
    \color{firebrick}
    \Precondition{$|A_i| < 2$ or $u(A_i - \last_1(A_i) - \last_2(A_i)) \le u(A_1)$ for all $i \in [n]$}
    \color{black}
    \For {$i \in [n]$}{
        Set $A_i$ as unmarked\;
    }
    \While {$\A$ has unmarked bundles} { 
        Let $i$ be the maximum index in $[n]$ such that $A_i$ is unmarked\; 
        \If {$i>1$ \emph{and} $|A_i| \ge 1$ \emph{and} $\Comp(A_1, A_i - \last(A_i))$ \label{line:sendmaxtomin_dash_compare}} {
            Remove $\last(A_i)$ from $A_i$ and append it to the end of $A_1$\label{line:sendmaxtomin_dash_transfer}\; 
            Set $A_1$ \textcolor{firebrick}{and $A_i$} as marked\label{line:sendmaxtomin_dash_mark}\;
            Reorder $\A$ to maintain $u(A_1) \le u(A_2) \le \dots \le u(A_n)$\label{line:sendmaxtomin_dash_reorder}\;
        }
        \Else {
            Set $A_i$ as marked\label{line:sendmaxtomin_dash_mark2}\;
        }
    }
    \Return $\A = (A_1, A_2, \dots, A_n)$\;
}
\Function{$\ScalingSendMaxToMin(u, n, \A = (A_1, A_2, \dots, A_n))$} {
    \Precondition{$u(A_1) \le u(A_2) \le \dots \le u(A_n)$}
    \If {$\sum_{i\in [n]} |A_i| \le 2n$} {
        \Return $\SendMaxToMin(u,n,\A)$\;
    }
    \For {$i \in [n]$}{
        $\mathcal{P}_i \gets$ Partition $A_i$ into adjacent pairs of goods; if $|A_i|$ is odd, leave one good unpaired\;
        $A_i' \gets$ Create a new bundle by replacing each pair in $\mathcal{P}_i$ with a single meta-good representing the pair; if there was an unpaired good, include it as is\;
    }
    $\A' \coloneqq (A_1', A_2', \dots, A_n')$\;
    $\A' \gets \ScalingSendMaxToMin(u, n, \A')$ \label{line:scalingsendmaxtomin_recurse}\;
    $\A \gets$ Expand each meta-good in $\A'$ into its (meta-)goods at the preceding recursion level, preserving their order\;
    $\A \gets \SendMaxToMin'(u, n, \A)$\;
    \Return $\A$\;
}
\end{algorithm}
    
We first show that $\SendMaxToMin'$ outputs a 1-witness EF1 allocation using $O(n \log n)$ comparison queries under the given precondition.

\begin{lemma}
\label{lem:dash-non-decreasing-min}
    During the execution of $\SendMaxToMin'$, $\min_{A \in \A} u(A)$ is non-decreasing.
\end{lemma}
\begin{proof}
    The lemma follows by the same argument as \Cref{lem:non-decreasing-min}.
\end{proof}

\begin{lemma}
\label{lem:dash-sendmaxtomin-is-witness-EF1}
    The output of $\SendMaxToMin'$ is a 1-witness EF1 allocation.
\end{lemma}
\begin{proof}
    Assuming the preconditions are satisfied, we prove that at the end of each iteration of the while loop, every marked bundle $X \in \A$ satisfies either $u(X) \le u(A_1)$ or $u(X - \last(X)) \le u(A_1)$.
    Since all bundles are marked at the termination, this implies that the output is 1-witness EF1.

    We first observe the following claim.
    \begin{claim}
        \label{clm:unmarked-bundles-meets-precondition}
        At the beginning of each iteration of the while loop in $\SendMaxToMin'$, 
        $|X| < 2$ or $u(X - \last_1(X) - \last_2(X)) \le u(A_1)$ holds for all unmarked bundles $X \in \A$.
    \end{claim}
    \begin{proof}
        At the beginning of the algorithm, the precondition ensures that $|X| < 2$ or $u(X - \last_1(X) - \last_2(X)) \le u(A_1)$ is satisfied for any bundle $X$.
        By \Cref{lem:dash-non-decreasing-min}, $u(A_1)$ is non-decreasing across iterations. 
        Thus, if $X$ has not been modified, then that inequality remains to be satisfied.
        Since any bundle that has sent or received a good is marked at line~\ref{line:sendmaxtomin_dash_mark}, we can conclude that at the beginning of each iteration, $|X| < 2$ or $u(X - \last_1(X) - \last_2(X)) \le u(A_1)$ holds for all unmarked bundles $X \in \A$.
    \end{proof}
    
    \UnderlineCase{Base case:} We show that in each iteration, when a bundle is first marked (either at line~\ref{line:sendmaxtomin_dash_mark} or line~\ref{line:sendmaxtomin_dash_mark2}), the condition of this lemma holds for the newly marked bundle at the end of that iteration.
    For a bundle marked after receiving a good at line~\ref{line:sendmaxtomin_dash_mark}, removing the appended good recovers the previous minimum value, which is not more than the current minimum.
    For a bundle $X$ marked at line~\ref{line:sendmaxtomin_dash_mark2}, either $u(X)$ is already the current minimum or $u(X-\last(X))$ is at most the current minimum.
    Thus, the condition holds for these cases.

    To see the remaining case, consider the case where an unmarked bundle $X$ transfers a good at lines~\ref{line:sendmaxtomin_dash_transfer}--\ref{line:sendmaxtomin_dash_reorder} and is marked for the first time.
    Let $a$ denote the value of $u(A_1)$ at the beginning of this iteration, and let $a'$ denote its value at the end.
    By \Cref{lem:dash-non-decreasing-min}, we have $a \le a'$.
    Furthermore, by \Cref{clm:unmarked-bundles-meets-precondition}, at the beginning of the iteration, $|X| < 2$ or $u(X - \last_1(X) - \last_2(X)) \le a$ holds.
    Let $X' = X-\last_1(X)$.
    After removing $\last_1(X)$ from $X$, we obtain $|X'| < 1$ or $u(X' - \last_1(X')) \le a \le a'$.
    Therefore, at the end of the iteration, $u(X') \le a'$ or $u(X' - \last_1(X')) \le a'$ holds.

    \UnderlineCase{Inductive step:} The proof is similar to that of \Cref{lem:sendmaxtomin-is-witness-EF1}.
\end{proof}

We show the query complexity of $\SendMaxToMin'$.

\begin{lemma}
\label{lem:dash-sendmaxtomin-complexity}
    $\SendMaxToMin'$ uses $O(n \log n)$ comparison queries.
\end{lemma}
\begin{proof}
    In each iteration of the while loop, at least one bundle is newly marked.
    Therefore, the total number of iterations is at most $n$.
    Each iteration uses one comparison query at line~\ref{line:sendmaxtomin_dash_compare} and $O(\log n)$ comparison queries at line~\ref{line:sendmaxtomin_dash_reorder}.
    Hence, the algorithm uses $O(n \log n)$ comparison queries in total.
\end{proof}

We show that $\ScalingSendMaxToMin$ outputs a 1-witness EF1 allocation.
Its query complexity is $O(n\log n\log(m/n))$.

\begin{lemma}
\label{lem:scaling-sendmaxtomin-is-witness-EF1}
    $\ScalingSendMaxToMin$ outputs a 1-witness EF1 allocation.
\end{lemma}
\begin{proof}
    We prove the lemma by induction on the recursion depth.

    \UnderlineCase{Base case:} 
    If the recursion stops because the number of (meta-)goods is at most $2n$, then we apply $\SendMaxToMin$.
    By \Cref{lem:sendmaxtomin-is-witness-EF1}, the resulting allocation is 1-witness EF1.

    \UnderlineCase{Inductive step:} Assume that $\A'$ at the end of line~\ref{line:scalingsendmaxtomin_recurse} is a 1-witness EF1 allocation.
    By the definition of 1-witness EF1, we have $|A_i'| < 1$ or $u(A_i' - \last(A_i')) \le u(A_1')$ for all $i \in [n]$.
    When $A_i$ is obtained by expanding each good in $A_i'$ back to the original pair of goods, whether $\last(A_i')$ represents a pair of goods or an unpaired good, we have $|A_i| < 2$ or $u(A_i - \last_1(A_i) - \last_2(A_i)) \le u(A_1') = u(A_1)$.
    Therefore, the precondition of $\SendMaxToMin'$ is satisfied, and by \Cref{lem:dash-sendmaxtomin-is-witness-EF1}, the output is a 1-witness EF1 allocation.
\end{proof}

\begin{lemma}
\label{lem:scaling-sendmaxtomin-complexity}
    Let $m = |M|$.
    Then $\ScalingSendMaxToMin$ uses $O(n \log n \log (m/n))$ comparison queries if $m> 2n$, and $O(n \log n)$ comparison queries otherwise.
\end{lemma}
\begin{proof}
    Since we stop the recursion immediately after the number of goods becomes at most $2n$, the recursion depth is $O(\log_2 (m/n))$.
    When at most $2n$ goods exist, we apply $\SendMaxToMin$, which uses $O(n \log n)$ comparison queries by \Cref{lem:sendmaxtomin-complexity}.
    At each level with more than $2n$ goods, $\SendMaxToMin'$ costs $O(n \log n)$ comparison queries by \Cref{lem:dash-sendmaxtomin-complexity}.
    This completes the proof.
\end{proof}

From \Cref{lem:scaling-sendmaxtomin-is-witness-EF1,lem:scaling-sendmaxtomin-complexity}, we can find an EF1 allocation by calling $\ScalingSendMaxToMin$ with a sorted initial allocation $\A = (\emptyset, \dots, \emptyset, M)$.

\begin{theorem}
\label{thm:sendmaxtomin}
    For identical and monotone valuations, there exists an algorithm that finds an EF1 allocation using $O(n \log n \log (m/n))$ comparison queries. 
\end{theorem}

We remark that each recursion level, which repairs a 1-witness EF1 using $O(n\log n)$ queries, contains a sorting-like component: deficient bundles must be matched with compensating goods while preserving the designated witness.
Appendix~\ref{app:witness-repair-lb} shows a lower bound of $\Omega(n\log n)$ on our repairing task.

\section{EF1+\texorpdfstring{$1/2$}{1/2}-TPS Allocation for Identical and Additive Valuations}

From this section onward, we assume additive valuations.
In this section, we strengthen our guarantee beyond EF1 by additionally ensuring a $1/2$-TPS approximation, which implies $1/2$-MMS.
More specifically, we prove that under identical additive valuations, an allocation that is simultaneously EF1, $1/2$-TPS and $1/2$-MMS can be computed using $O(n\log n\log (m/n))$ comparison queries (\Cref{thm:ef1-halfmms}).

To provide the additional $1/2$-TPS guarantee, we need to avoid any allocations in which a high-value good is a witness of EF1 and the minimum value of bundles is far from the TPS value.
A key observation to overcome this situation is that, by computing a 1-witness EF1 allocation, we can enumerate all high-value goods.
For a partition $\A$ of $M$ with $u(A_1)\le u(A_2) \le \dots \le u(A_k)$, we denote $W(\A) = \{\last(A_i) \mid 2 \le i \le k \text{ and } A_i \ne \emptyset\}$.
\begin{lemma}
\label{lem:large-goods-in-witnesses}
    Let $\A = (A_1, A_2, \dots, A_k)$ be a 1-witness EF1 allocation under identical and additive valuations with $u(A_1) \le u(A_2) \le \dots \le u(A_k)$.
    If a good $g \in M$ satisfies $u(g) > \min_{A \in \A} u(A)$, then $g \in W(\A)$.
    In particular, $W(\A)$ contains all the goods $g$ with $u(g) > u(M) / k$.
\end{lemma}
\begin{proof}
    Fix a good $g$ satisfying $u(g) > u(A_1)$, and let $A_i \ (\neq A_1)$ be the bundle containing $g$.
    Note that $u(A_1) < u(g) \le u(A_i)$.
    If $\last(A_i) \ne g$, then we would have $u(A_1) < u(g) \le u(A_i - \last(A_i))$ by the monotonicity of $u$, which contradicts the 1-witness EF1 property.
    Therefore, $g \in W(\A)$.
    The last statement follows since $u(A_1) \leq u(M)/k$.
\end{proof}

Leveraging \Cref{lem:large-goods-in-witnesses}, we construct the algorithm as follows.
The formal description is given in~\Cref{alg:ef1-halfmms}.
First, we identify the good set $W_0$ containing goods whose values exceed the proportional share by using $\ScalingSendMaxToMin$ for $2n+1$ agents.
Second, we compute an EFX allocation $\A$ of $W_0$.
We then append all the remaining goods to the end of the bundle with the highest value, and apply the following variant of $\ScalingSendMaxToMin$ (\Cref{alg:scaling-sendmaxtomin}), which we call $\ProtectedScalingSendMaxToMin$.
A (meta-)good is called protected if all the original goods represented by it belong to $W_0$.
Throughout the execution, the protected (meta-)goods precede others in each bundle.
At every recursion level, we pair adjacent protected goods and adjacent other goods respectively, preserving their order.
If the number of protected goods is odd, then we leave the last protected good unpaired.
Whenever $\SendMaxToMin$ or $\SendMaxToMin'$ would transfer a protected good, we cancel the transfer and mark the source bundle.

To obtain an EFX allocation of $W_0$, we apply the Egal-Sequential algorithm~\cite{AR2020}.
This algorithm sorts all goods in non-increasing order of value, and then sequentially assigns each good to the bundle with the currently smallest total value.
In the following lemma, the EFX guarantee follows from \cite[Lemma~1]{AR2020}, and the query bound follows from $O(m\log m)$ comparisons for sorting the goods and $O(m\log n)$ comparisons for the subsequent binary insertions.
\begin{lemma}
\label{lem:egal-sequential-EFX}
    For $n$ agents with an identical additive valuation and $m$ goods, the Egal-Sequential algorithm outputs an EFX allocation using $O(m \log (n+m))$ comparison queries.
\end{lemma}
In \Cref{alg:ef1-halfmms}, lines~\ref{line:ef1-halfmms-es-start}--\ref{line:ef1-halfmms-es-end} correspond to the Egal-Sequential algorithm.
Since $|W_0|\leq 2n$, an EFX allocation of $W_0$ can be found using $O(n \log n)$ comparison queries.

\begin{algorithm}[htb]
    \caption{Computing an EF1+$1/2$-TPS allocation under identical and additive valuation}
    \label{alg:ef1-halfmms}
    \Function {$\EFOneHalfMMS(u, n, M)$}{
        Compute a 1-witness EF1 allocation of length $2n + 1$ by 
        $\B \gets \ScalingSendMaxToMin(u, 2n + 1, (\emptyset, \dots, \emptyset, M))$\; \label{line:ef1-halfmms-compute-B}
        Let $W_0 \gets \{\last(B_i) \mid 2 \le i \le 2n + 1 \text{ and } B_i \ne \emptyset\}$ be the set of witnesses\; \label{line:ef1-halfmms-set-W}
        \Comment{Egal-Sequential allocation of $W_0$ to $n$ agents}
        Initialize allocation $\A = (A_1, A_2, \dots, A_n)$ with $A_i = \emptyset$ for all $i \in [n]$\;\label{line:ef1-halfmms-es-start}
        Sort goods in $W_0 = \{w_1, w_2, \dots, w_{|W_0|}\}$ to get $u(w_1) \ge u(w_2) \ge \dots \ge u(w_{|W_0|})$\; \label{line:ef1-halfmms-sort}
        \For {$g = w_1, w_2, \dots, w_{|W_0|}$ \label{line:ef1-halfmms-main-loop}}{
            Append $g$ to the end of $A_1$\;
            Reorder $\A$ to maintain $u(A_1) \le u(A_2) \le \dots \le u(A_n)$\; \label{line:ef1-halfmms-reorder}
            \Comment{Use binary insertion with $O(\log n)$ comparison queries} \label{line:ef1-halfmms-es-end}
        }
        Append $M \setminus W_0$ to the end of $A_n$\; \label{line:ef1-halfmms-append}
        \Return $\ProtectedScalingSendMaxToMin(u, n, \A, W_0)$\; \label{line:ef1-halfmms-scalingsendmaxtomin}
        \Comment{Ensure that each good of $W_0$ remains in its assigned bundle during \ScalingSendMaxToMin}
    }
\end{algorithm}

We now establish the correctness of \Cref{alg:ef1-halfmms}.

\begin{lemma}
\label{lem:ef1-halfmms-non-decreasing-min}
    During the execution of \Cref{alg:ef1-halfmms}, $\min_{A \in \A} u(A)$ is non-decreasing.
\end{lemma}
\begin{proof}
    We see that $\min_{A \in \A} u(A)$ is non-decreasing during lines~\ref{line:ef1-halfmms-main-loop}--\ref{line:ef1-halfmms-append}.
    During the execution of $\ProtectedScalingSendMaxToMin$ at line~\ref{line:ef1-halfmms-scalingsendmaxtomin}, every executed transfer of a good does not decrease $\min_{A \in \A} u(A)$ by \Cref{lem:non-decreasing-min,lem:dash-non-decreasing-min}, and a canceled transfer leaves the allocation unchanged.
\end{proof}

\begin{lemma}
\label{lem:ef1-halfmms-fix-W}
    $\ProtectedScalingSendMaxToMin$ outputs a 1-witness EF1 allocation such that every good in $W_0$ remains in the bundle to which it was assigned by the Egal-Sequential algorithm.
\end{lemma}
\begin{proof}
    Let $\A'$ denote the allocation at the end of lines~\ref{line:ef1-halfmms-main-loop}--\ref{line:ef1-halfmms-reorder} in which the goods in $W_0$ are allocated to $n$ agents using the Egal-Sequential algorithm.
    By \Cref{lem:egal-sequential-EFX}, this allocation is EFX. 
    Let $\beta \coloneqq \min_{i\in[n]}u(A'_i)$.

    By construction, any meta-good represents either only goods in $W_0$ or goods in $M\setminus W_0$.
    Since the algorithm cancels any transfer of protected goods, at any point, each current bundle $A_j \in \A$ contains some bundle $A'_{j'} \subseteq W_0$ obtained by the Egal-Sequential algorithm.
    Thus, $\min_{j\in [n]} u(A_j) \geq \min_{j\in [n]} u(A'_j) = \beta$.

    Suppose that a transfer from a bundle $A_j$ is canceled.
    The last meta-good $\last(A_j)$ is protected.
    Since protected goods precede others in each bundle, $A_j$ has only protected goods.
    Thus, when viewed as a set of original goods, $A_j$ coincides with some bundle $A'_{j'}$ in $\A'$.
    Moreover, $\last(A'_{j'}) \in W_0$ is represented by the meta-good $\last(A_j)$ and $\A'$ is EFX.
    Hence, we have $u(A_j-\last(A_j))\le u(A'_{j'}-\last(A'_{j'}))\le \beta \le \min_{A \in \A} u(A)$ by \Cref{lem:ef1-halfmms-non-decreasing-min}.
    Thus, we can mark $A_j$ using $\last(A_j)$ as its witness, and continue the execution.

    The canceled transfer is the only additional case in the argument of marked bundles. Every executed transfer is handled exactly as in $\ScalingSendMaxToMin$.
    The proof of \Cref{lem:scaling-sendmaxtomin-is-witness-EF1} applies after adding the cancellation case above.
    Therefore, we can compute a 1-witness EF1 allocation while keeping every good in $W_0$ in its initially assigned bundle.
\end{proof}

The following corollary is an immediate consequence of the proof of \Cref{lem:ef1-halfmms-fix-W}.

\begin{corollary}
    \label{cor:W-precedes-non-W}
    In the output of \Cref{alg:ef1-halfmms}, elements of $W_0$ precede elements not in $W_0$ within each bundle, and each element in $W_0$ remains in the bundle to which it was assigned by the Egal-Sequential algorithm, keeping the non-increasing order of value.
\end{corollary}

This leads to the following structural properties of high-value goods.

\begin{lemma}
\label{lem:large-goods-must-be-singleton}
    Let $\A$ denote the output of \Cref{alg:ef1-halfmms}. 
    If a good $g \in M$ satisfies $u(g) > u(M) / n$, then there exists a bundle $X \in \A$ such that $X = \{g\}$.
\end{lemma}
\begin{proof}
    Let $g$ be a good with $u(g) > u(M) / n$, and let $A_i \in \A$ be the bundle containing $g$.
    
    By applying \Cref{lem:large-goods-in-witnesses} to $\B$ computed at line~\ref{line:ef1-halfmms-compute-B}, the set $W_0=W(\B)$ contains every good whose value is greater than $u(B_1)$.
    Since $u(B_1)\le {u(M)}/(2n+1)<{u(M)}/{n}<u(g)$, we have $g\in W_0$.
    Moreover, $\last(A_i) = g$ because otherwise $u(A_i - \last(A_i)) \ge u(g) > u(M)/n \ge u(A_1)$, which contradicts \Cref{lem:ef1-halfmms-fix-W}.
    By \Cref{cor:W-precedes-non-W}, all remaining elements in $A_i$ also belong to $W_0$.
    Suppose for contradiction that $|A_i| \ge 2$. 
    By the Egal-Sequential procedure, we have $u(\last_2(A_i)) \ge u(\last(A_i)) = u(g) > u(M) / n$.
    Therefore, $u(A_i - \last(A_i)) > u(M) / n \ge u(A_1)$, which contradicts the fact that $\A$ is a 1-witness EF1 allocation (by~\Cref{lem:ef1-halfmms-fix-W}).
    Hence, we must have $A_i = \{g\}$.
\end{proof}

\begin{lemma}
\label{lem:no-singleton-implies-halfprop}
    Let $W \subseteq M$ be any set of goods containing every good $g$ with $u(g) > \frac{1}{2} u(M)/n$.
    Let $\A=(A_1,\ldots,A_n)$ be an allocation of $M$ satisfying
    \begin{itemize}
        \item $u(A_1) \le u(A_2) \le \dots \le u(A_n)$,
        \item $\A$ is 1-witness EF1, and contains no singleton $X$ with $u(X) > u(M)/n$, and 
        \item in every bundle $A_i$, all goods in $W$ precede all goods outside $W$, and the goods in $W\cap A_i$ are ordered in non-increasing order of value.
    \end{itemize}
    Then $u(A_1) \ge u(M) / 2n$.
\end{lemma}
\begin{proof}
    Assume for contradiction that $u(A_1) < \frac{1}{2} u(M) / n$.
    
    First, since $u(A_1)<\frac{1}{2} u(M)/n$ and $A_n$ has the maximum value in $\A$, we must have $u(A_n) > u(M)/n$.
    By the assumption of singletons, $|A_n| \ge 2$, and $\last_2(A_n)$ exists.
    Next, since $\A$ is 1-witness EF1, we have $u(A_n - \last(A_n)) \le u(A_1)$.
    Combining this with $u(A_1) < \frac{1}{2} u(M) / n$ and $u(A_n) \ge u(M) / n$, we obtain $u(A_n - \last(A_n)) < \frac{1}{2} u(M) / n < u(\last(A_n))$.
    Hence, $\last(A_n) \in W$.

    Since all goods in $W$ precede all goods outside $W$ in each bundle, $A_n$ contains only goods in $W$, and particularly $\last_2(A_n) \in W$.
    Moreover, as the goods in $A_n \cap W$ are ordered in non-increasing order of value, we have $u(\last_2(A_n)) \ge u(\last(A_n))$.
    However, this contradicts $u(\last_2(A_n)) \leq u(A_n - \last(A_n)) < u(\last(A_n))$.
\end{proof}

Combining the above two lemmas, we see that if the output $\A$ of \Cref{alg:ef1-halfmms} contains a singleton, then we can find the one with a high-value good.
This is useful in the later section.

\begin{corollary}
\label{cor:largest-singleton-is-half-prop}
    Let $\A$ denote the output of \Cref{alg:ef1-halfmms}. 
    If there exists a bundle $X \in \A$ satisfying $|X| = 1$, then there exists a bundle $X \in \A$ satisfying $|X| = 1$ and $u(X) \ge \frac{1}{2} u(M) / n$.
\end{corollary}
\begin{proof}
    Let $X \in \A$ be a singleton bundle with maximum value.
    Assume for contradiction that $u(X) < \frac{1}{2} u(M) / n$.
    We note that there exists no singleton bundle with value greater than $u(M) / n$.
    Thus, $\A$ and $W_0$ at line~\ref{line:ef1-halfmms-set-W} satisfy the condition of \Cref{lem:no-singleton-implies-halfprop} by \Cref{lem:ef1-halfmms-fix-W,lem:large-goods-in-witnesses}.
    We can apply \Cref{lem:no-singleton-implies-halfprop} to obtain $u(X) \ge u(A_1) \ge \frac{1}{2} u(M) / n$, which is a contradiction.
\end{proof}

Now we are ready to show the guarantee of EF1 and $1/2$-TPS.

\begin{lemma}
\label{lem:ef1-halfmms-is-ef1-halfmms}
    \Cref{alg:ef1-halfmms} outputs an allocation that is simultaneously 1-witness EF1, $1/2$-TPS, and $1/2$-MMS.
\end{lemma}
\begin{proof}
    Let $\A = (A_1, A_2, \dots, A_n)$ denote the output of \Cref{alg:ef1-halfmms}, sorted so that $u(A_1)\le \cdots \le u(A_n)$.
    By \Cref{lem:ef1-halfmms-fix-W}, $\A$ satisfies 1-witness EF1.

    Let $\B=(B_1,\ldots,B_{2n+1})$ be the allocation computed at line~\ref{line:ef1-halfmms-compute-B}, and let $W_0=W(\B)$ be the set computed at line~\ref{line:ef1-halfmms-set-W}.
    \Cref{lem:large-goods-in-witnesses} for $\B$ implies that $W_0$ contains every good $g$ with $u(g) > u(M)/(2n+1)$.

    We now iteratively remove singleton bundles containing a high-value good from $\A$.
    Let $(n^0,M^0,\A^0)=(n,M,\A)$.
    At step $t$, if $\A^t$ contains a singleton bundle $X^t=\{x_t\}$ such that
    \begin{align}
        u(X^t)>\frac{u(M^t)}{n^t}, \label{eq:singleton}
    \end{align}
    then set $n^{t+1}=n^t-1$, $M^{t+1}=M^t\setminus\{x_t\}$, and $\A^{t+1}=\A^t\setminus\{X^t\}$.
    If no such singleton exists, or if $n^t=1$ or $u(M^t) = 0$, then the process stops.

    We claim that at each step $t$ with $u(M^t)>0$, the set $W^t \coloneqq W_0\cap M^t$ contains every good $g\in M^t$ with $u(g)>{u(M^t)}/(2n^t)$.
    Consider the bundles of $\B$ that contain no good of $M\setminus M^t$.
    There are at least $2n+1-t$ such bundles, and their total value is at most
    $u(M^t)$. Since $B_1$ has the minimum value, we have
    \begin{align*}
        u(B_1)\le \frac{u(M^t)}{2n+1-t} < \frac{u(M^t)}{2n^t}.
    \end{align*}
    Therefore, every $g\in M^t$ with $u(g)>{u(M^t)}/{(2n^t)}$ satisfies $u(g)>u(B_1)$, and hence belongs to $W_0$ by \Cref{lem:large-goods-in-witnesses}.

    The process stops at some step $\ell < n$ by construction.
    Throughout the process, $A_1$ is never removed and remains a bundle with minimum value at any point.
    If $n^\ell=1$ or $u(M^\ell)=0$, we have $u(A^\ell_1) \ge u(M^\ell)/n^\ell$.
    Suppose otherwise.
    By construction, $\A^\ell$ contains no singleton $X$ with $u(X)>{u(M^\ell)}/n^\ell$.
    In addition, because we only remove a singleton $X^t$ satisfying \eqref{eq:singleton} from $\A$, the following three properties are inherited from the original allocation: (a) $\A^t$ is a 1-witness EF1 allocation of $M^t$, sorted in non-decreasing order of value; (b) $A_1$ remains the minimum-value bundle in $\A^t$; and (c) in every bundle of $\A^t$, all goods in $W^t$ precede all goods outside $W^t$, and the goods in $W^t$ are ordered in non-increasing order of value.
    Thus, $\A^\ell$ and $W^\ell$ satisfy the assumptions of \Cref{lem:no-singleton-implies-halfprop}, and it holds that
    \begin{align*}
        u(A_1) = \min_{A\in \A^\ell}u(A)
        \ge \frac{u(M^\ell)}{2n^\ell}
        \ge \frac{1}{2} \TPS^{n^\ell}(M^\ell).
    \end{align*}
    Therefore, in all cases, $u(A_1) \ge \frac{1}{2} {u(M^\ell)}/{n^\ell} \ge \frac{1}{2} \TPS^{n^\ell}(M^\ell)$ holds, where the last inequality follows from $\TPS^{n^\ell}(M^\ell)\le u(M^\ell)/n^\ell$.
    By repeatedly applying \Cref{lem:TPS-singleton}, we have $u(A_1)\ge \frac{1}{2} \TPS^n(M)$.
    Since $A_1$ has the minimum value among $\A$, every bundle of $\A$ has value at least $u(A_1)$.
    Thus, $\A$ satisfies $1/2$-TPS, and hence $1/2$-MMS.
\end{proof}

Combining \Cref{lem:ef1-halfmms-is-ef1-halfmms} and the analysis of the query complexity, we obtain the following theorem.

\begin{theorem}
\label{thm:ef1-halfmms}
    For identical and additive valuations, 
    \Cref{alg:ef1-halfmms} computes an EF1, $1/2$-TPS, and $1/2$-MMS allocation using $O(n \log n \log (m/n))$ comparison queries.
\end{theorem}
\begin{proof}
    We execute \Cref{alg:ef1-halfmms} for $(u, n, M)$.
    The output is EF1, $1/2$-TPS, and $1/2$-MMS by~\Cref{lem:ef1-halfmms-is-ef1-halfmms}.
    We analyze the query complexity.
    Line~\ref{line:ef1-halfmms-compute-B} uses $O(n \log n \log (m/n))$ comparison queries by \Cref{lem:scaling-sendmaxtomin-complexity}.
    Since $|W_0| \le 2n$ (line~\ref{line:ef1-halfmms-set-W}), lines~\ref{line:ef1-halfmms-sort}--\ref{line:ef1-halfmms-reorder} use $O(n \log n)$ comparison queries.

    It remains to analyze $\ProtectedScalingSendMaxToMin$ at line~\ref{line:ef1-halfmms-scalingsendmaxtomin}.
    After $t$ contraction steps, there are at most $|W_0|/2^t+n$ protected meta-goods and at most $|M\setminus W_0|/2^t+1$ meta-goods from $M\setminus W_0$.
    Hence, the recursion depth is $O(\log_2 (m/n))$.
    Cancellation of a transfer requires no additional comparison query, and we use $O(n \log n \log (m/n))$ comparison queries by the argument in \Cref{lem:scaling-sendmaxtomin-complexity}.
    Therefore, the total number of comparison queries used by \Cref{alg:ef1-halfmms} is $O(n \log n \log (m/n))$.
\end{proof}


\subsection{Restricted Non-Identical EF1 using Witness Detection}

In this subsection, we demonstrate that the witness information produced by the algorithm for identical valuations is useful in a special class of additive valuations which are not necessarily identical.
We may assume that $m>n$ because assigning at most one good to each agent gives a 1-witness EF1 allocation.
The key ingredient is \Cref{lem:large-goods-in-witnesses}.
By this result, if sufficiently high-value goods exist for the agents, then we can extract such goods using $\ScalingSendMaxToMin$.
Such goods help us compress the remaining goods into a small number of bundles, and we can apply the standard round-robin algorithm for the meta-goods to find an EF1 allocation.
The proofs in this subsection are deferred to Appendix~\ref{app:non-identical-EF1-proof}.

\begin{proposition}\label{prop:non-identical-EF1-heavy-anchors}
    Suppose that valuations are additive.
    If there exist pairwise distinct goods $h_1,\ldots,h_n\in M$ such that $u_i(h_i) > u_i(M)/(nP)$ for every $i\in [n]$ for some integer $P \ge 1$, then a 1-witness EF1 allocation can be computed using  $O(\mathrm{poly}(n, P)\log m)$ comparison queries.
\end{proposition}

The assumption is existential, as neither the high-value goods $h_1,\dots, h_n$ nor the integer $P$ is given to the algorithm.

Furthermore, a similar trick works when valuations are additive and binary (i.e., $u_i(g)\in\{0,1\}$ for every agent $i$ and good $g$), and the agents' supports (the set of positive-value goods) overlap arbitrarily, unlike \Cref{prop:non-identical-EF1-heavy-anchors}.
\begin{proposition}\label{prop:non-identical-EF1-binary}
    Suppose that valuations are binary additive.
    Let $S_i\coloneqq \{g\in M\mid u_i(g)=1\}$ for each $i\in [n]$, and let $K\coloneqq \max_{i\in [n]} |S_i|$.
    If $\min_{i\in [n]} |S_i| \geq 1$, then a 1-witness EF1 allocation can be computed using $O(\mathrm{poly}(n,K)\log m)$ comparison queries.
    In particular, if $K=\mathrm{poly}(n, \log m)$, the query complexity is $O(\mathrm{poly}(n,\log m))$.
\end{proposition}

We remark that the binary assumption in \Cref{prop:non-identical-EF1-binary} is used to discover the support scale without being given $K$.
Since the comparison query cannot determine whether $u_i(X)=0$, the algorithm cannot detect that the number of bundles has exceeded the support sizes.


\section{PROP1+\texorpdfstring{$1/2$}{1/2}-TPS Allocation for Non-Identical and Additive Valuations}

In this section, we turn to the general setting, where each agent $i$ has her own additive valuation $u_i$ and valuations may \emph{not} be identical.
In the comparison-based query model, this setting is substantially more challenging, since comparisons are agent-specific and information cannot be freely shared across agents.
The main result of this section is to provide an $O(n^3 \log m)$-query algorithm whose output is simultaneously PROP1 and $1/2$-TPS (hence $1/2$-MMS).
The construction is based on two consequences from the identical case.
First, our algorithms for the identical case provide a compact ``certificate'' to decide quickly whether a bundle is PROP1 and non-PROP for each agent.
By utilizing this observation, we can accelerate the matching-based framework of Bu, Li, Liu, Song, and Tao~\cite{BLLST2024}.
Second, \Cref{alg:ef1-halfmms} exposes every good whose value exceeds the proportional share as a singleton bundle.
We can ensure that such goods are allocated without harming the TPS values of the remaining agents.
In the following, we first present a simple PROP1 algorithm (\Cref{alg:prop1-non-identical}), and then further strengthen this to guarantee $1/2$-TPS (\Cref{alg:prop1-half-mms-non-identical}) within the same query complexity.

We begin by introducing the notion of 1-witness PROP1 allocations. 
When valuations are additive, any 1-witness EF1 allocation is also 1-witness PROP1.

\begin{definition}
    An allocation $\A = (A_1, A_2, \dots, A_n)$ is said to satisfy \emph{1-witness PROP1} if for all agents $i \in [n]$, either $u_i(A_i) \ge u_i(M) / n$ or there exists a nonempty bundle $A_j \in \A \setminus \{A_i\}$ such that $u_i(A_i + \last(A_j)) \ge u_i(M) / n$.
\end{definition}

\begin{lemma}
\label{lem:1-witness-ef1-implies-1-witness-prop1}
    When valuations are additive, any 1-witness EF1 allocation is also a 1-witness PROP1 allocation.
\end{lemma}
\begin{proof}
    Let $\A = (A_1, A_2, \dots, A_n)$ be a 1-witness EF1 allocation.
    Fix an arbitrary agent $i \in [n]$.
    If $u_i(A_i)\ge u_i(M)/n$, then we are done.
    Otherwise, a bundle $A_j$ maximizing $u_i(A_j)$ is nonempty and satisfies $j\ne i$.
    In this case, by the definition of 1-witness EF1, we have $u_i(A_i + \last(A_j)) \ge u_i(A_j) \ge u_i(M) / n$.
    Therefore, $\A$ is a 1-witness PROP1 allocation.
\end{proof}

\subsection{Matching-Based Framework}
\subsubsection{Subroutine \PROPOneNonPROP}

We introduce the $\PROPOneNonPROP$ subroutine proposed in~\cite{BLLST2024}.
Given a bundle $A_1$ and a valuation function $u$, this subroutine returns one of the following based on the evaluation of bundle $A_1$:
\begin{itemize}
    \item If $u(A_1) > \frac{1}{n} u(M)$, it returns $\true$.
    \item If  $u(A_1) < \frac{1}{n} u(M)$ and $u(A_1 + g) < \frac{1}{n} u(M)$ for all goods $g \in M \setminus A_1$, it returns $\false$.
    \item Otherwise, it may return either $\true$ or $\false$.
\end{itemize}
In other words, if $\PROPOneNonPROP$ returns $\true$, it is guaranteed that $A_1$ satisfies PROP1, and if it returns $\false$, it is guaranteed that $u(A_1) \le u(M) / n$.

Roughly speaking, $\PROPOneNonPROP$ arranges the goods on a path, and repeatedly finds a minimal contiguous bundle along the path whose value exceeds that of the given bundle.
See~\Cref{alg:propone-nonprop} for the formal description.

\begin{algorithm}[htb]
\caption{Determining whether bundle $A_1$ is PROP1 or non-PROP under additive valuation}
\label{alg:propone-nonprop}
\Function {$\PROPOneNonPROP(u, n, M, A_1)$}{
    Let $R \gets M \setminus A_1$ be the remaining goods\;
    \For {$i = 2$ \emph{to} $n$} {
        \lIf {$\Comp(R, A_1) = \true$ or $R = \emptyset$} {let $(D_i, I_i) \gets (R, \emptyset)$}
        \lElse{
            Use binary search to find a contiguous bundle $D_i\cup I_i \subseteq R$ such that $(\Comp(D_i, A_1) = \true \ \text{or} \ D_i=\emptyset)$ and $\Comp(D_i \cup I_i, A_1) = \false$, where $I_i$ is a singleton
        }
        $R \gets R \setminus (D_i \cup I_i)$\;
    }
    \lIf {$R$ is empty} {\Return $\true$}
    \lElse {\Return $\false$}
}
\end{algorithm}

It is shown that $\PROPOneNonPROP$ has the following property.

\begin{lemma}[\cite{BLLST2024}]
\label{lem:prop1-nonprop}
    If \Cref{alg:propone-nonprop} returns $\true$ then $A_1$ satisfies PROP1, and otherwise, $A_1$ satisfies $u(A_1) \le u(M) / n$.
    Moreover, \Cref{alg:propone-nonprop} uses $O(n \log m)$ comparison queries.
\end{lemma}

\subsubsection{Want-This Graph and Hall Matching}

The PROP1 algorithm of Bu, Li, Liu, Song, and Tao~\cite{BLLST2024} is a matching-based framework: it repeatedly (1) let some agent $k$ divide unallocated goods into bundles satisfying PROP1 for agent $k$, (2) constructs a bipartite graph between unallocated agents and candidate bundles and selects a non-empty matching via the Dulmage--Mendelsohn decomposition, and (3) assign bundles according to the matching.
We refer to the bipartite graph as the \emph{want-this graph}.
More specifically, the want-this graph is a bipartite graph having agents and bundles as vertices, satisfying the following conditions:
\begin{itemize}
    \item If there is an edge between agent $i$ and bundle $X$, then assigning bundle $X$ to agent $i$ ensures that agent $i$ satisfies PROP1.
    \item If there is no edge between agent $i$ and bundle $X$, then even if bundle $X$ is assigned to another agent, there exists an alternative way to satisfy PROP1 for agent $i$. 
\end{itemize}
In~\cite{BLLST2024}, the output of $\PROPOneNonPROP$ directly determined the presence or absence of edges for unallocated goods every time a bundle is assigned to some agent.
We will describe the idea to accelerate the construction of the want-this graph through preprocessing and binary search in the next subsection.

The following lemma from~\cite{BLLST2024} forms the theoretical foundation for the PROP1 allocation algorithm.

\begin{lemma}[\cite{BLLST2024}]\label{lem:hall-matching}
    Let $(L \cup R, E)$ be a bipartite graph with $n + n$ vertices, and suppose there exists a left vertex $k \in L$ such that $(k, r) \in E$ for all $r \in R$.
    Then there exists a non-empty matching $H \subseteq E$ satisfying the following property:
    \begin{quote}
        For every edge $(l, r) \in E$, if $l$ is unmatched in $H$, then $r$ is also unmatched in $H$.
    \end{quote}
\end{lemma}
Following~\cite{BLLST2024}, we call a matching satisfying this condition a \emph{Hall matching}.

A Hall matching is crucial for the algorithm of~\cite{BLLST2024}.
By the construction of a want-this graph, in a Hall matching, agents receive bundles satisfying PROP1, or have nothing but do not value the matched bundles more than their own PROP values.
Thus, we may assign some bundles according to the Hall matching without harming the possibility of receiving a PROP1 bundle for the remaining agents.

We present the algorithm $\HallMatching$ to find a Hall matching in~\Cref{alg:hall-matching}.
This computes a maximum-cardinality matching, and then constructs a Hall matching based on the Dulmage--Mendelsohn decomposition.

\begin{algorithm}[ht]
    \caption{Computing a Hall matching in a bipartite graph}
    \label{alg:hall-matching}
    \Function {$\HallMatching(L \cup R, E, k)$}{
        \Precondition{$|L| = |R|$ and $(k, r) \in E$ for all $r \in R$}
        Let $H_m$ be a maximum-cardinality matching in $(L \cup R, E)$\;
        \If {$H_m$ is a perfect matching} {
            \Return $H_m$\;
        }
        \Comment{Find the Dulmage--Mendelsohn decomposition of the graph}
        Let $U_L \subseteq L$ be the set of unmatched left vertices in $H_m$\;
        Let $Z_L \subseteq L$ be the set of left vertices reachable from $U_L$ via alternating paths
        (i.e., paths that alternate between edges in $E \setminus H_m$ and edges in $H_m$)\;
        Let $Z_R \subseteq R$ be the set of right vertices reachable from $U_L$ via alternating paths\;
        Let $H' \gets \{ (l, r) \in H_m \mid l \in Z_L \}$\;
        \Return $H_m \setminus H'$\;
    }
\end{algorithm}

Since \Cref{alg:hall-matching} can be implemented using one maximum-cardinality matching computation and one Breadth-First-Search, the following complexity result holds.

\begin{lemma}
\label{lem:hall-matching-complexity}
    \Cref{alg:hall-matching} runs in $O(T_{\text{BM}}(|L| + |R|, |E|))$ time, where $T_{\text{BM}}(n, m)$ denotes the time complexity of computing a maximum-cardinality matching in a bipartite graph with $n$ vertices and $m$ edges.
\end{lemma}
For dense graphs $G=(V, E)$, the state-of-the-art result~\cite{CK2024} implies $T_{\text{BM}}(|V|,|E|) = O(|V|^{2 + o(1)})$.

\subsection{Warm-up: Algorithm for PROP1 Allocation}

We now present our PROP1 algorithm as a warm-up.
We summarize our algorithm in~\Cref{alg:prop1-non-identical}.
Since the cases $n=1$ and $M=\emptyset$ are trivial, we assume that $n\ge 2$ and $M\ne\emptyset$.
Similarly to~\cite{BLLST2024}, we repeatedly do the following: for the remaining good set $P$ and agent set $Q$, compute a 1-witness EF1 allocation $\B$ under some agent $k$'s valuation, construct a Hall matching $H$ on a want-this graph, and assign bundles to agents according to $H$, keeping every unallocated agent's proportional share for the remaining goods at least the original one.
A want-this graph is constructed using $\PROPOneNonPROP$ for the currently unallocated goods in every iteration.
The key observation behind our speedup is that it suffices to know the answer of $\PROPOneNonPROP$ with respect to the original instance.
We accelerate the algorithm by determining the answer of $\PROPOneNonPROP$ without executing it in most cases as follows.
We precompute a 1-witness PROP1 allocation $\mathcal{C}$ under each agent's valuation, and learn $T_i \coloneqq C_1$ and the witness good $g_i$ as the certificate of PROP1 for the agent $i$.
Instead of executing $\PROPOneNonPROP$, we just compare the target bundle with $C_1$.
The exception is when the target bundle contains the PROP1 witness for $C_1$, and in this case, we call $\PROPOneNonPROP$.
Here, it can happen that no bundle is ``better than $T_i$.''
However, we will show that in this case, every bundle is PROP.

\begin{algorithm}[htb]
\caption{Computing a PROP1 allocation for non-identical and additive valuations}
\label{alg:prop1-non-identical}
\Function {${\sf PROP1}([n], M)$}{
    \Comment{Preprocessing: Find thresholds and witnesses}
    \For {$i \in [n]$}{\label{line:prop1-preprocess-start}
        Compute a 1-witness PROP1 allocation for agent $i$ by
        $\mathcal{C} \gets \ScalingSendMaxToMin(u_i, n, (\emptyset, \dots, \emptyset, M))$\;\label{line:prop1-define-Ti-gi-start}
        \lIf {$\mathcal{C}$ contains an empty bundle}{Relabel the bundles so that $C_1$ is empty}
        \Comment{$T_i$ serves as the threshold for non-PROP bundles}
        Let $T_i \gets C_1$\;
        \Comment{$g_i$ serves as a witness that any bundle with value at least $u(T_i)$ satisfies PROP1}
        Let $g_i \in \{\last(C_j) \mid j \in [n]\setminus \{1\} \text{ and } C_j \ne \emptyset\}$ that maximizes $u_i(g_i)$\;\label{line:prop1-define-Ti-gi-end}
    }
    \Comment{Main loop}
    Let $P \gets M$ be the set of unallocated goods\;
    Let $Q \gets [n]$ be the set of unallocated agents\;
    Initialize allocation $\A = (A_1, A_2, \dots, A_n)$ with $A_i \gets \emptyset$ for all $i \in [n]$\;
    \While{$Q$ is not empty} {
        Choose an arbitrary agent $k \in Q$\;
        Compute a PROP1 allocation for agent $k$ by
        $\B \gets \ScalingSendMaxToMin(u_k, |Q|, (\emptyset, \dots, \emptyset, P))$\;\label{line:prop1-compute-B}
        \Comment{Construct the want-this graph $(Q \cup \B, E)$}
        Let $E \gets \{ (k, B_j) \mid B_j \in \B \}$\;\label{line:prop1-graph-start}
        \ForEach {$i\in Q \setminus \{k\}$} {
            Let $F_i$ be a bundle maximizing $u_i(B)$ over $B \in \B$\;\label{line:prop1-set-Bmax}
            Add edge $(i, F_i)$ to $E$\;
            \If {$\Comp_i(F_i, T_i)$} {
                Add all edges $\{(i, B_j) \mid B_j \in \B\}$ to $E$\;
                \textbf{continue}\;\label{line:prop1-continue}
            }
            \ForEach {$B_j \in \B \setminus \{F_i\}$} {
                \lIf {$\Comp_i(B_j, T_i)$} {\textbf{continue}}\label{line:prop1-comp}
                \If {$g_i \notin B_j$ \emph{or} $\PROPOneNonPROP(u_i, |Q|, P, B_j)$} {\label{line:prop1-nonprop}
                    Add edge $(i, B_j)$ to $E$\;
                }
            }
        }\label{line:prop1-graph-end}
        \ForEach {$(i, B_j) \in \HallMatching(Q \cup \B, E, k)$} {
            \Comment{Assign bundle $B_j$ to agent $i$}
            Update $A_i \gets B_j$, $Q \gets Q \setminus \{i\}$, and $P \gets P \setminus B_j$\;\label{line:prop1-assign}
        }
    }
    \Return $\A = (A_1, A_2, \dots, A_n)$\;
}
\end{algorithm}
\clearpage

We now establish the correctness of \Cref{alg:prop1-non-identical}.

\begin{lemma}
\label{lem:fast-prop1-nonprop}
    The values $T_i$ and $g_i$ computed in lines~\ref{line:prop1-define-Ti-gi-start}--\ref{line:prop1-define-Ti-gi-end} of \Cref{alg:prop1-non-identical} satisfy the following properties:
    \begin{itemize}
        \item For any bundle $X$, if $\Comp_i(X,T_i)=\true$, then $u_i(X)\le u_i(M)/n$.
        \item For any bundle $X$, if $g_i \notin X$ and $\Comp_i(X, T_i) = \false$, then $X$ satisfies PROP1 for agent $i$, that is, either $u_i(X) \ge u_i(M) / n$ or $u_i(X + g) \ge u_i(M) / n$ for some good $g \in M \setminus X$.
    \end{itemize}
\end{lemma}
\begin{proof}
    \underline{\emph{For the first property:}} Since $\mathcal{C}$ satisfies $u_i(C_1) \le u_i(C_2) \le \dots \le u_i(C_n)$, we have $u_i(T_i) = u_i(C_1) \le u_i(M) / n$.
    Therefore, if $\Comp_i(X, T_i) = \true$, then $u_i(X) \le u_i(T_i) \le u_i(M) / n$.

    \UnderlineCase{For the second property:} Since $\mathcal{C}$ is a 1-witness PROP1 allocation by \Cref{lem:1-witness-ef1-implies-1-witness-prop1}, if $g_i \notin X$ and $\Comp_i(X, T_i) = \false$, then $u_i(X + g_i) \ge u_i(T_i + g_i) = u_i(C_1 + g_i) \ge u_i(M) / n$, which implies that bundle $X$ satisfies PROP1 for agent $i$.
\end{proof}

\begin{lemma}
\label{lem:increasing-prop}
    At each iteration of the while loop in \Cref{alg:prop1-non-identical}, we have $u_i(P) / |Q| \ge u_i(M) / n$ for any $i \in Q$.
    Furthermore, the want-this graph $(Q \cup \B, E)$ constructed in lines~\ref{line:prop1-graph-start}--\ref{line:prop1-graph-end} satisfies 
    $u_i(B_j) \le u_i(P) / |Q|$ for any non-edge $(i, B_j) \not\in E$.
\end{lemma}
\begin{proof}
    We prove by induction on the number of iterations.
    
    \UnderlineCase{Base case:} 
    In the first iteration, we have $(Q, P) = ([n], M)$, so clearly $u_i(P) / |Q| = u_i(M) / n$ for any $i \in Q$.
    For the want-this graph, if $(i, B_j) \not\in E$, then either $\Comp_i(B_j, T_i) = \true$ or $\PROPOneNonPROP(u_i, |Q|, P, B_j) = \false$.
    If $\Comp_i(B_j, T_i) = \true$, then by \Cref{lem:fast-prop1-nonprop}, we have $u_i(B_j) \le u_i(M) / n = u_i(P) / |Q|$.
    If $\PROPOneNonPROP(u_i, |Q|, P, B_j) = \false$, then by \Cref{lem:prop1-nonprop}, we have $u_i(B_j) \le u_i(P) / |Q|$.
    Thus, in either case, $(i, B_j) \not\in E$ implies $u_i(B_j) \le u_i(P) / |Q|$.

    \UnderlineCase{Inductive step:}
    Let $P', Q', E', \B'$ denote the values of $P, Q, E, \B$ in the previous iteration, respectively, and let $H'$ denote the Hall matching computed in the previous iteration.
    By the inductive hypothesis, for any $i \in Q'$, we had $u_i(P') / |Q'| \ge u_i(M) / n$, and for any non-edge $(i, B_j') \not\in E'$, we had $u_i(B_j') \le u_i(P') / |Q'|$.

    Fix an arbitrary agent $i \in Q$. For any $(i', B_j') \in H'$, by \Cref{lem:hall-matching}, we have $(i, B_j') \not\in E'$, which implies $u_i(B_j') \le u_i(P') / |Q'|$.
    Since only bundles with value at most the average are removed, we have $u_i(P) / |Q| \ge u_i(P') / |Q'| \ge u_i(M) / n$.
    The want-this graph property follows by the same argument as in the base case.
\end{proof}

\begin{lemma}
\label{lem:want-this-graph-true}
    The want-this graph $(Q \cup \B, E)$ constructed in lines~\ref{line:prop1-graph-start}--\ref{line:prop1-graph-end} of \Cref{alg:prop1-non-identical} satisfies the following:
    \begin{quote}
        For any edge $(i, B_j) \in E$, bundle $B_j$ satisfies PROP1 for agent $i$.
        That is, either $u_i(B_j) \ge u_i(M) / n$ or there exists a good $g \in M \setminus B_j$ such that $u_i(B_j + g) \ge u_i(M) / n$.
    \end{quote}
\end{lemma}
\begin{proof}
    An edge $(i, B_j) \in E$ is added in one of the following cases:
    \begin{itemize}
        \item $i = k$.
        \item $B_j = F_i$.
        \item $\Comp_i(F_i, T_i) = \true$.
        \item $\Comp_i(B_j, T_i) = \false$ and $g_i \notin B_j$.
        \item $\PROPOneNonPROP(u_i, |Q|, P, B_j) = \true$.
    \end{itemize}

    \UnderlineCase{When $i = k$:} The allocation $\B = (B_1, B_2, \dots, B_{|Q|})$ is a PROP1 allocation with respect to $(u_i, |Q|, P)$.
    Therefore, each bundle $B_j \in \B$ satisfies either $u_i(B_j) \ge u_i(P) / |Q|$ or there exists a good $g \in P \setminus B_j$ such that $u_i(B_j + g) \ge u_i(P) / |Q|$.
    By \Cref{lem:increasing-prop}, we have $u_i(P) / |Q| \ge u_i(M) / n$, and thus $B_j$ satisfies PROP1 with respect to $(u_i, n, M)$ as well.

    \UnderlineCase{When $B_j = F_i$:} Since $F_i$ has the maximum value among the $|Q|$ bundles that partition $P$, we have $u_i(F_i) \ge u_i(P)/|Q| \ge u_i(M)/n$, where the second inequality follows from \Cref{lem:increasing-prop}.
    Hence $F_i$ is PROP for agent $i$.

    \UnderlineCase{When $\Comp_i(F_i,T_i) = \true$:} This condition implies that $u_i(F_i) \le u_i(T_i)$.
    Since $\mathcal{C}$ satisfies $u_i(C_1) \le u_i(C_2) \le \dots \le u_i(C_n)$, we have $u_i(T_i) \le u_i(M)/n$.
    Moreover, by the previous case, we also have $u_i(F_i) \ge u_i(P)/|Q| \ge u_i(M)/n$.
    Thus, it follows that $u_i(F_i) = u_i(T_i) = u_i(M)/n$ and $u_i(P)/|Q|=\,u_i(M)/n$.
    Since every bundle in $\B$ has value at most $u_i(F_i)$ and the bundles in $\B$ partition $P$, every bundle in $\B$ has value exactly $u_i(M)/n$.
    Therefore, $B_j$ is PROP.
    
    \UnderlineCase{When $\Comp_i(B_j, T_i) = \false$ and $g_i \notin B_j$:} By \Cref{lem:fast-prop1-nonprop}, bundle $B_j$ satisfies the PROP1 condition.

    \UnderlineCase{When $\PROPOneNonPROP(u_i, |Q|, P, B_j) = \true$:} By \Cref{lem:prop1-nonprop}, bundle $B_j$ satisfies PROP1 with respect to $(u_i, |Q|, P)$.
    By \Cref{lem:increasing-prop}, we have $u_i(P) / |Q| \ge u_i(M) / n$, and thus $B_j$ satisfies PROP1 with respect to $(u_i, n, M)$ as well.
\end{proof}

Summing up the above lemmas, we obtain the following theorem.
\begin{theorem}
\label{thm:prop1-non-identical}
    For non-identical and additive valuations, 
    \Cref{alg:prop1-non-identical} computes a PROP1 allocation using $O(n^3 \log m)$ comparison queries.
\end{theorem}
\begin{proof}
    Execute \Cref{alg:prop1-non-identical} and let $\A$ be the output.
    
    Bundles are assigned to agents at line~\ref{line:prop1-assign}, where $(i, B_j) \in E$ holds.
    By \Cref{lem:want-this-graph-true}, the assigned bundle $B_j$ satisfies PROP1 for agent $i$.
    Therefore, the final output allocation $\A$ is a PROP1 allocation.

    The preprocessing phase (lines~\ref{line:prop1-preprocess-start}--\ref{line:prop1-define-Ti-gi-end}) uses $O(n^2 \log n \log m)$ comparison queries.
    We now analyze the main loop.
    Since the Hall matching is non-empty by~\Cref{lem:hall-matching}, at least one agent is assigned a bundle in each iteration.
    Therefore, the while loop executes at most $n$ iterations.
    In each iteration, line~\ref{line:prop1-compute-B} uses $O(n\log n\log m)$ comparison queries to compute $\B$, lines~\ref{line:prop1-set-Bmax}--\ref{line:prop1-continue} use $O(n)$ comparison queries per agent, line~\ref{line:prop1-comp} performs one comparison query for each of $O(n^2)$ pairs, and line~\ref{line:prop1-nonprop} invokes $\PROPOneNonPROP$ for at most $n$ pairs.
    Recall that $\PROPOneNonPROP$ uses $O(n \log m)$ comparison queries by \Cref{lem:prop1-nonprop}.
    Overall, the bottleneck is $O(n^2)$ calls to $\PROPOneNonPROP$, yielding $O(n^3 \log m)$ comparison queries in total.
\end{proof}

\subsection{Additional \texorpdfstring{$1/2$}{1/2}-TPS Guarantee}

In this subsection, we add the TPS guarantee into the approach of \Cref{alg:prop1-non-identical}.
We present an algorithm that outputs a PROP1 and $1/2$-TPS allocation in $O(n^3\log m)$ comparison queries 
(\Cref{thm:prop1-half-mms-non-identical}).
When $|M|<n$, an allocation in which each good is assigned to a distinct agent is 1-witness EF1, and hence PROP1, while both TPS and MMS values are zero.
Thus, we may assume that $|M|\ge n$.

To guarantee $1/2$-TPS, we need to be careful about high-value goods, particularly those whose value exceeds the proportional share for some agent.
Our algorithm, which is formally described in \Cref{alg:prop1-half-mms-non-identical}, adopts $\EFOneHalfMMS$ (\Cref{alg:ef1-halfmms}) instead of $\ScalingSendMaxToMin$.
This allows us to learn each agent's PROP1 certificate and identify such high-value goods.
If such a good is exposed, then we immediately assign the good as a singleton bundle to the agent.
This assignment does not harm any other agent as long as we focus on PROP1 and $1/2$-TPS.
These modifications preserve the query complexity while ensuring the $1/2$-TPS (hence $1/2$-MMS) guarantee.

\begin{algorithm}[htb]
    \caption{Computing a PROP1+$1/2$-TPS allocation for non-identical and additive valuations}
    \label{alg:prop1-half-mms-non-identical}
    \Function {$\text{\sf PROP1+$1/2$-TPS}([n], M)$} {
        \Precondition{$|M| \ge n$}
        Let $P \gets M$ be the set of unallocated goods\;
        Let $Q \gets [n]$ be the set of unallocated agents\;
        Initialize allocation $\A = (A_1, A_2, \dots, A_n)$ with $A_i \gets \emptyset$ for all $i \in [n]$\;
        \Comment{Preprocessing: Find thresholds and witnesses}
        \For {$i \in [n]$}{\label{line:prop1-half-mms-preprocessing}
            Compute a 1-witness EF1+$1/2$-TPS allocation for agent $i$ by
            $\mathcal{C} \gets \EFOneHalfMMS(u_i, |Q|, P)$\;\label{line:prop1-half-mms-compute-C}
            \If {$\mathcal{C}$ contains a singleton} {\label{line:prop1-half-mms-check-singleton-prep}
                \Comment{Assign $C_j$ to agent $i$; Otherwise, no good has value exceeding $u_i(P) / |Q|$}
                Let $C_j$ be a singleton with maximum value\;
                Update $A_i \gets C_j$, $Q \gets Q \setminus \{i\}$, and $P \gets P \setminus C_j$\;\label{line:prop1-half-mms-assign-singleton-prep}
                \textbf{continue}\;
            }
            \lIf {$\mathcal{C}$ contains an empty bundle}{Relabel the bundles so that $C_1$ is empty}
            \Comment{$T_i$ serves as the threshold for non-$1/2$-TPS bundles}
            Let $T_i \gets C_1$\; \label{line:prop1-half-mms-define-Ti}
            \Comment{$g_i$ serves as a witness that any bundle with value at least $u(T_i)$ satisfies PROP1}
            \If {$|Q|>1$}{
                Let $g_i \in \{\last(C_j) \mid j \in [|Q|]\setminus \{1\} \text{ and } C_j \ne \emptyset\}$ that maximizes $u_i(g_i)$\; \label{line:prop1-half-mms-define-gi}
            }
        }\label{line:prop1-half-mms-preprocessing-end}
        \Comment{Main loop}
        \While{$Q$ is not empty\label{line:prop1-half-mms-main-loop}} {
            Choose an arbitrary agent $k \in Q$\;
            Compute a 1-witness EF1+$1/2$-TPS allocation for agent $k$ by
            $\B \gets \EFOneHalfMMS(u_k, |Q|, P)$\;\label{line:prop1-half-mms-compute-B}
            \If {$\B$ contains a singleton (i.e., bundle with size 1)} {
                \Comment{Assign $B_j$ to agent $k$; otherwise, all bundles have value at least $u_k(P) / 2|Q|$}
                Let $B_j$ be a singleton with maximum value\;\label{line:prop1-half-mms-find-singleton}
                Update $A_k \gets B_j$, $Q \gets Q \setminus \{k\}$, and $P \gets P \setminus B_j$\;\label{line:prop1-half-mms-assign-singleton-main}
                \textbf{continue}\;
            }
            \Comment{Construct the want-this graph $(Q \cup \B, E)$}
            Let $E \gets \{ (k, B_j) \mid B_j \in \B \}$\;
            \ForEach {$i \in Q \setminus \{k\}$} {
                Let $F_i$ be a bundle maximizing $u_i(B)$ over $B \in \B$\;\label{line:prop1-half-mms-find-max}
                Add edge $(i, F_i)$ to $E$\;
                \If {$\Comp_i(F_i, T_i)$} {
                    Add all edges $\{(i, B_j) \mid B_j \in \B\}$ to $E$\;
                    \textbf{continue}\;\label{line:prop1-half-mms-add-all-edges}
                }
                \ForEach {$B_j \in \B \setminus \{F_i\}$} {
                    \lIf {$\Comp_i(B_j, T_i)$} {\textbf{continue}}\label{line:prop1-half-mms-compare}
                    \If {$g_i \notin B_j$ \emph{or} $\PROPOneNonPROP(u_i, |Q|, P, B_j)$} {\label{line:prop1-half-mms-check-witness}
                        Add edge $(i, B_j)$ to $E$\;
                    }
                }
            }
            \ForEach {$(i, B_j) \in \HallMatching(Q \cup \B, E, k)$} {
                \Comment{Assign bundle $B_j$ to agent $i$}
                Update $A_i \gets B_j$, $Q \gets Q \setminus \{i\}$, and $P \gets P \setminus B_j$\;\label{line:prop1-half-mms-assign-bundle}
            }
        }
        \Return $\A = (A_1, A_2, \dots, A_n)$\;
    }
\end{algorithm}

We now establish the correctness of \Cref{alg:prop1-half-mms-non-identical}.
\newcommand{\prop}{\mathit{prop}}
For each agent $i$, let $P_i$ and $Q_i$ denote the sets $P$ and $Q$ at the beginning of the preprocess (line~\ref{line:prop1-half-mms-preprocessing}) for agent $i$, and we denote 
$$\prop_i \coloneqq \frac{u_i(P_i)}{|Q_i|}.$$
Agents assigned singleton bundles in the preprocess are excluded from $Q$ in the main loop (line~\ref{line:prop1-half-mms-main-loop}).
For every other agent $i$, let $T_i$ be the bundle chosen at line~\ref{line:prop1-half-mms-define-Ti}, and if $|Q_i|>1$ then let $g_i$ be the good chosen at line~\ref{line:prop1-half-mms-define-gi}.

\begin{lemma}
\label{lem:average-is-more-than-prop}
    Fix agent $i \in [n]$.
    At each iteration of the main loop from line~\ref{line:prop1-half-mms-main-loop}, if $i \in Q$, then $u_i(P) / |Q| \ge \prop_i$.
\end{lemma}
\begin{proof}
    Suppose $i\in Q$, which implies that the condition at line~\ref{line:prop1-half-mms-check-singleton-prep} is not satisfied for agent $i$.
    By \Cref{lem:large-goods-must-be-singleton}, 
    $u_i(T_i)\le \prop_i$, and no good has value exceeding $\prop_i$.
    Therefore, whenever a singleton is assigned to another agent, the inequality $u_i(P) / |Q| \ge \prop_i$ is preserved.
    
    Thus, it suffices to show that the assignments at line~\ref{line:prop1-half-mms-assign-bundle} also preserve $u_i(P) / |Q| \ge \prop_i$.
    When a bundle $B_j$ is assigned to some agent $i' \ne i$, the Hall matching property gives $(i, B_j) \notin E$, which implies that either $\Comp_i(B_j, T_i) = \true$ or $\PROPOneNonPROP(u_i, |Q|, P, B_j) = \false$.
    If $\PROPOneNonPROP(u_i, |Q|, P, B_j) = \false$, then $u_i(B_j) \le u_i(P) / |Q|$, and thus $u_i(P) / |Q| \ge \prop_i$ is maintained after the assignment.
    If $\Comp_i(B_j, T_i) = \true$, then $u_i(B_j) \le u_i(T_i) \le \prop_i$, and hence $u_i(P) / |Q| \ge \prop_i$ is again preserved.
\end{proof}

\begin{lemma}
\label{lem:prop1-and-half-mms-for-current-PQ}
    If bundle $A_i$ is assigned to agent $i$, then $A_i$ satisfies the following properties:
    \begin{itemize}
        \item $u_i(A_i) \ge \prop_i / 2$,
        \item $u_i(A_i) \ge \prop_i$ or $u_i(A_i + g) \ge \prop_i$ for some good $g \in P_i \setminus A_i$.
    \end{itemize}
\end{lemma}
\begin{proof}
    Note that by \Cref{lem:average-is-more-than-prop}, we have $u_i(P) / |Q| \ge \prop_i$ also throughout the main loop.
    Bundles are assigned to agents at lines~\ref{line:prop1-half-mms-assign-singleton-prep},~\ref{line:prop1-half-mms-assign-singleton-main},~and~\ref{line:prop1-half-mms-assign-bundle}.

    \UnderlineCase{Bundle $C_j$ at line~\ref{line:prop1-half-mms-assign-singleton-prep}:}
    The first property follows from \Cref{cor:largest-singleton-is-half-prop}.
    Since $\mathcal{C}$ is a PROP1 allocation by \Cref{lem:1-witness-ef1-implies-1-witness-prop1}, $C_j$ satisfies the second property.

    \UnderlineCase{Bundle $B_j$ at line~\ref{line:prop1-half-mms-assign-singleton-main}:} The argument is analogous to the case of line~\ref{line:prop1-half-mms-assign-singleton-prep}.

    \UnderlineCase{Bundle $B_j$ at line~\ref{line:prop1-half-mms-assign-bundle} when $i = k$:} 
    The first property follows from \Cref{lem:no-singleton-implies-halfprop}.
    Since $\B$ is a PROP1 allocation, $B_j$ satisfies the second property.

    \UnderlineCase{Bundle $B_j$ at line~\ref{line:prop1-half-mms-assign-bundle} when $i \ne k$:} 
    Since $i\in Q\setminus\{k\} \subseteq Q_i \setminus \{k\}$, we have $|Q_i| \ge |Q|>1$, and hence $g_i$ is well-defined.
    By construction, the allocation $\mathcal{C}$ computed at line~\ref{line:prop1-half-mms-compute-C} for agent $i$ does not contain singletons.
    By \Cref{lem:no-singleton-implies-halfprop}, $T_i$ satisfies $u_i(T_i) \ge \prop_i / 2$.
    The edge $(i,B_j)$ is added in one of the following cases.
    If $B_j=F_i$, then we have $u_i(B_j)\ge u_i(P)/|Q|\ge \prop_i$, and hence both properties hold for $B_j$.
    Suppose that the edge is added at line~\ref{line:prop1-half-mms-add-all-edges}.
    We have $u_i(F_i)\le u_i(T_i)$.
    Since $u_i(T_i)\le \prop_i$ and $u_i(F_i)\ge u_i(P)/|Q|\ge \prop_i$, we also have $u_i(F_i)= u_i(T_i) = \prop_i$.
    As $F_i$ has the maximum value, every bundle in $\B$ has value exactly $\prop_i$ for agent $i$.
    Hence, both properties hold for $B_j$.
    Finally, suppose that the edge $(i,B_j)$ is added at line~\ref{line:prop1-half-mms-check-witness}.
    From the comparison at line~\ref{line:prop1-half-mms-compare}, we have $u_i(B_j) \ge u_i(T_i) \ge \prop_i / 2$, and thus the first property holds.
    For the second property, if $\Comp_i(B_j, T_i) = \false$ and $g_i \notin B_j$, then the second property follows from the 1-witness PROP1 property of $\mathcal{C}$ computed at line~\ref{line:prop1-half-mms-compute-C}.
    If $\PROPOneNonPROP(u_i, |Q|, P, B_j) = \true$, then $B_j$ satisfies the PROP1 condition with respect to $(u_i, |Q|, P)$ by \Cref{lem:prop1-nonprop}, and since $u_i(P) / |Q| \ge \prop_i$, the second property holds.
\end{proof}

\begin{lemma}
\label{lem:prop1-half-mms-non-identical-correctness}
    \Cref{alg:prop1-half-mms-non-identical} outputs a PROP1, $1/2$-TPS, and $1/2$-MMS allocation.
\end{lemma}
\begin{proof}
    Let $\A$ be the output of~\Cref{alg:prop1-half-mms-non-identical} for $(M, [n])$.
    Fix an arbitrary agent $i \in [n]$.
    
    Note that any bundle assigned to other agents before the preprocess of agent $i$ is a singleton.
    Repeatedly applying \Cref{lem:TPS-singleton} to these singletons gives $\TPS_i^{|Q_i|}(P_i) \ge \TPS_i^n(M)$.
    Since $\TPS_i^{|Q_i|}(P_i) \le u_i(P_i)/|Q_i| = \prop_i$, we have $\prop_i \ge \TPS_i^n(M)$.
    By \Cref{lem:prop1-and-half-mms-for-current-PQ}, bundle $A_i \in \A$ has value at least $\prop_i/2$.
    Thus, $A_i$ is $1/2$-TPS with respect to $(M, [n])$ for agent $i$.
    
    It remains to check PROP1.
    The second property of \Cref{lem:prop1-and-half-mms-for-current-PQ} shows that $A_i$ satisfies the PROP1 condition with respect to $(P,Q)$ at the moment that $\prop_i$ is defined.
    Repeatedly applying \Cref{lem:prop1-singleton} to the singletons assigned to other agents before the preprocess for agent $i$, we conclude that $A_i$ also satisfies PROP1 with respect to the original instance $(M,[n])$.
    
    Therefore, $\A$ satisfies PROP1, $1/2$-TPS, and $1/2$-MMS with respect to $(M, [n])$.
\end{proof}

\begin{lemma}
\label{lem:prop1-half-mms-non-identical-complexity}
    \Cref{alg:prop1-half-mms-non-identical} uses $O(n^3 \log m)$ comparison queries.
\end{lemma}
\begin{proof}
    The preprocessing phase (lines~\ref{line:prop1-half-mms-preprocessing}--\ref{line:prop1-half-mms-preprocessing-end}) uses $O(n^2 \log n \log m)$ comparison queries by~\Cref{thm:ef1-halfmms}.
    For the main loop, since the Hall matching is non-empty by~\Cref{lem:hall-matching}, at least one agent is assigned a bundle in each iteration.
    Therefore, the while loop executes at most $n$ iterations.
    In each iteration, line~\ref{line:prop1-half-mms-compute-B} performs $O(n \log n \log m)$ comparison queries to compute $\B$, lines~\ref{line:prop1-half-mms-find-singleton}, \ref{line:prop1-half-mms-find-max}--\ref{line:prop1-half-mms-add-all-edges}, and~\ref{line:prop1-half-mms-compare} use $O(n^2)$ comparison queries in total, and line~\ref{line:prop1-half-mms-check-witness} invokes $\PROPOneNonPROP$ for at most $n$ pairs.
    Overall, the bottleneck is $O(n^2)$ calls to $\PROPOneNonPROP$, yielding $O(n^3 \log m)$ comparison queries in total.
\end{proof}

The following theorem follows from \Cref{lem:prop1-half-mms-non-identical-correctness,lem:prop1-half-mms-non-identical-complexity}.

\begin{theorem}
\label{thm:prop1-half-mms-non-identical}
    For non-identical and additive valuations, 
    \Cref{alg:prop1-half-mms-non-identical} computes a PROP1, $1/2$-TPS, and $1/2$-MMS allocation using $O(n^3 \log m)$ comparison queries.
\end{theorem}


\section{\texorpdfstring{$k$}{k}-witness EF1 and its relation to TPS and MMS}

In this section we introduce a stronger variant of 1-witness EF1, called \emph{$k$-witness EF1}, and relate it to $\alpha$-TPS or $\alpha$-MMS guarantees under additive valuations.
The output of $\EFOneHalfMMS$ (\Cref{alg:ef1-halfmms}) for identical valuations satisfies a slightly stronger fact than just being 1-witness EF1: for each bundle $A_j$ other than $A_1$, if it contains only goods in $W_0$ set at line~\ref{line:ef1-halfmms-set-W}, then any good in $A_j$ is a witness for EF1.
Thus, we ask what the notion of $k$-witness EF1 would imply if it were available.
More specifically, using a clean structural connection to TPS and MMS of $k$-witness EF1, we prove that every 2-witness EF1 allocation guarantees $1/2$-TPS, and every 3-witness EF1 allocation guarantees $4/7$-MMS (\Cref{thm:2-witness-implies-half-mms,thm:3-witness-implies-47-mms}) in the general additive setting.
Here, the TPS guarantee cannot be improved for $3$-witness EF1 or even EFX allocations.

\subsection{\texorpdfstring{$k$}{k}-witness EF1}

A $k$-witness EF1 allocation is one in which, for each bundle, the last $k$ goods (if they exist) all serve as witnesses for the EF1 property.
\begin{definition}
For an integer $k \ge 1$, an allocation $\A = (A_1, A_2, \dots, A_n)$ is said to satisfy \emph{$k$-witness EF1} if for all agents $i, j \in [n]$ and all integers $1 \le k' \le k$, either $|A_j| < k'$ or $u_i(A_i) \ge u_i(A_j - \last_{k'}(A_j))$.
\end{definition}
Since an $m$-witness EF1 allocation is equivalent to an EFX allocation, $k$-witness EF1 can be viewed as an intermediate notion between EF1 and EFX.
For additive valuations, whether EFX allocations always exist is a central open question in fair division.
At this point, the existence is known for three agents~\cite{CGM24} and for instances with at most two~\cite{Mah23} or three~\cite{HGNV25} types of valuations, while relaxations such as EFX with bounded charity~\cite{CKMS21,AACGMM25}, almost-full EFX for four agents~\cite{BCFF22}, and graphical valuations~\cite{CFKS23} have also been studied.
Recently, an EFX allocation was shown \emph{not} to exist in general for monotone~\cite{AMMSW26} and with a compact construction, for submodular valuations~\cite{MS26}.
Hence, the additive case is now the frontier.
The $k$-witness hierarchy offers a relaxation orthogonal to these papers, as it parameterizes the number of goods per bundle that are certified to witness EF1.

We mention that for identical additive valuations, the \EgalSequential\ algorithm~\cite{AR2020} is guaranteed to output an EFX allocation, which is equivalent to an $m$-witness EF1 allocation.
Thus, even when $m>n$, for every $k\ge 1$, we can find a $k$-witness EF1 allocation with $O(m \log m)$ comparison queries, by sorting the goods, appending a good to a minimum-value bundle, and reordering bundles via binary insertion.
In contrast, whether a $2$-witness EF1 allocation can be computed with $o(m)$ comparison queries is open.

\subsection{Relation Between \texorpdfstring{$k$}{k}-witness EF1, TPS, and MMS}

We now investigate the structural implications of $k$-witness EF1 under general additive valuations.
While $k$-witness EF1 was introduced as a strengthened, certificate-based variant of EF1, it also yields nontrivial maximin share guarantees.
In this subsection we prove two implications of 2-witness and 3-witness EF1, and also provide upper bounds on the approximation ratios.

\begin{theorem}
\label{thm:2-witness-implies-half-mms}
    For additive valuations, any $2$-witness EF1 allocation $\A = (A_1, A_2, \dots, A_n)$ also satisfies $1/2$-TPS and $1/2$-MMS.
\end{theorem}
\begin{proof}
    Fix an agent $i$ arbitrarily.
    We show that $u_i(A_i) \ge \TPS^n_i(M)/2$.
    We proceed by induction on $n$.

    \UnderlineCase{Base case:}
    When $n = 1$, the allocation trivially satisfies $1/2$-TPS.

    \UnderlineCase{Inductive step:}
    If $u_i(A_i) =\max_{j'\in [n]} u_i(A_{j'})$, then $u_i(A_i) \ge u_i(M)/n \ge \TPS_i^n(M)$ and we are done.
    Otherwise, let $A_j$ be a bundle such that $u_i(A_j) = \max_{j'\in [n]}u_i(A_{j'}) > u_i(A_i) (\ge 0)$.
    We note that $A_j \ne A_i$, $|A_j| > 0$, and $u_i(A_j) \ge u_i(M)/n$.

    Let $\A' = (A_1, A_2, \dots, A_{j-1}, A_{j+1}, \dots, A_{n})$.
    This remains $2$-witness EF1.
    Assume that $\A'$ satisfies $1/2$-TPS, that is, $u_i(A_i) \ge \frac{1}{2}\TPS^{n-1}_i(M \setminus A_j)$.

    If $|A_j| = 1$, then by \Cref{lem:TPS-singleton}, we have $u_i(A_i) \ge \frac{1}{2}\TPS^{n-1}_i(M \setminus A_j) \ge \frac{1}{2}\TPS^{n}_i(M)$,  so the allocation satisfies $1/2$-TPS.

    Now consider the case $|A_j| \ge 2$. By the definition of $2$-witness EF1, the last two goods $x, y$ of $A_j$ satisfy $u_i(x) \le u_i(A_j - y) \le u_i(A_i)$ and $u_i(A_j - x) \le u_i(A_i)$.
    Combining $u_i(x) \le u_i(A_i)$ and $u_i(A_j - x) \le u_i(A_i)$, we obtain $u_i(A_j) \le 2 \cdot u_i(A_i)$.
    Therefore,
    \begin{align*}
        u_i(A_i) \ge \frac{1}{2} u_i(A_j) \ge \frac{1}{2n} u_i(M) \ge \frac{1}{2} \TPS^n_i(M),
    \end{align*}
    which implies that $\A$ satisfies $1/2$-TPS for agent $i$.
    The $1/2$-MMS guarantee also holds since $\TPS^n_i(M) \ge \mu^n_i(M)$.
\end{proof}

We note that \cite{BEF22} gives an instance with identical valuations and $2n-1$ goods of value $1$ such that, for every allocation, some agent receives at most an $\frac{n}{2n-1}$ fraction of the TPS value.
Thus, even EFX cannot guarantee more than $\frac{n}{2n-1}$-TPS, which tends to $1/2$ as $n \to \infty$.
By this observation, we focus on the relationship between $3$-witness EF1 and MMS.

\begin{theorem}
\label{thm:3-witness-implies-47-mms}
    For additive valuations, any $3$-witness EF1 allocation $\A = (A_1, A_2, \dots, A_n)$ also satisfies $4/7$-MMS.
\end{theorem}
\begin{proof}
    Fix an agent $i$ arbitrarily.
    Suppose that $u_i(A_i)=0$. 
    Since $\A$ is 3-witness EF1, for any bundle $A_j$ with $|A_j| \ge 2$, since $0\le u_i(\last_2(A_j)) \le u_i(A_j-\last_1(A_j)) \le u_i(A_i)$ and $0\le u_i(A_j-\last_2(A_j)) \le u_i(A_i)$, we have $u_i(A_j)=0$. 
    Thus, only singleton bundles can have positive value, and there are fewer than $n$ goods with positive value.
    This implies $\mu_i^n(M) = 0$ and the claim holds.
    
    In the remainder, without loss of generality, assume $u_i(A_i)=1$.
    The following three operations on $\A$ do not decrease $\mu^n_i(M)$, while preserving the $3$-witness EF1 property between $A_i$ and other bundles of size at least three:

    \begin{enumerate}
        \item If $|A_j| \le 1$ for some $j \ne i$, remove bundle $A_j$ and replace 
        $(n, M, \A) \gets (n - 1, M \setminus A_j, \A \setminus \{A_j\})$.
        \item If $|A_j| = 2$ for some $j \ne i$, let $x, y$ denote its two elements.
            If $\min\{u_i(x), u_i(y)\} < 1$, increase the values of these goods so that $u_i(x) = u_i(y) = 1$.
        \item If the number of bundles satisfying $j \ne i$ and $|A_j| = u_i(A_j) = 2$ exceeds $n/2$, select one such bundle $A$ and remove it by replacing $(n, M, \A) \gets (n - 1, M \setminus A, \A \setminus \{A\})$.
    \end{enumerate}

    We verify that each operation does not decrease $\mu^n_i(M)$.
    The first operation does not decrease $\mu^n_i(M)$ by \Cref{lem:MMS-singleton}, where the case $A_j = \emptyset$ is immediate.
    The second operation does not decrease $\mu^n_i(M)$.
    We note that, since values of goods in size-two bundles are changed only by the second operation, the 3-witness EF1 property holds for $A_i$ and the current bundle $A_j$.
    Hence, $\max\{u_i(x),u_i(y)\} \le u_i(A_i)=1$ holds before the operation.
    For the third operation, let $\B$ denote a partition of $M$ that achieves $\mu^n_i(M)$ before the removal in the current instance.
    Since there are more than $n$ goods $g \in M$ with value $1$ before the removal, there exists a bundle $B \in \B$ containing at least two such goods.
    Since we can exchange labels of goods with value $1$, we can view $\B \setminus \{B\}$ as a partial allocation of $M \setminus A$ to $n - 1$ agents, and we have
    \begin{align*}
        \mu^{n-1}_i(M \setminus A) \ge \min_{B' \in \B \setminus \{B\}} u_i(B') \ge \min_{B' \in \B} u_i(B') = \mu^n_i(M).
    \end{align*}
    Thus, the new MMS value is at least the old MMS value.

    Repeatedly applying whichever operation is applicable, we eventually reach a state where $\A$ admits none of the three operations.
    At this point, $\A$ satisfies the following properties:
    \begin{itemize}
        \item There is no bundle $A_j$ with $j \ne i$ and $|A_j| \le 1$.
        \item There are at most $n/2$ bundles $A_j$ with $j \ne i$ and $|A_j| = 2$.
        \item If $j \ne i$ and $|A_j| = 2$, then $u_i(A_j) = 2$.
    \end{itemize}
    Moreover, for any bundle $A_j$ with $|A_j| \ge 3$, let $x, y, z$ denote the last three goods.
    By the definition of $3$-witness EF1, we have $u_i(A_j - x) \le 1$, $u_i(A_j - y) \le 1$, and $u_i(A_j - z) \le 1$.
    Summing these inequalities yields
    \begin{align*}
        u_i(A_j) \le \frac{3}{2}u_i(A_j - x - y - z) + u_i(x) + u_i(y) + u_i(z) \le \frac{3}{2}.
    \end{align*}
    Since each bundle of size two contributes more to $u_i(M)$ than each bundle of size at least three, we obtain
    \begin{align*}
        u_i(M) \le 1 + \frac{n}{2} \cdot 2 + \left(\frac{n}{2} - 1\right) \cdot \frac{3}{2} < \frac{7n}{4},
    \end{align*}
    which implies
    \begin{align*}
        u_i(A_i) > \frac{4}{7} \cdot \frac{u_i(M)}{n} \ge \frac{4}{7} \cdot \mu^n_i(M).
    \end{align*}

    Since all three operations preserve the property that $\mu^n_i(M)$ does not decrease, and leave $A_i$ unchanged, the original allocation satisfies $4/7$-MMS.
\end{proof}

In addition to TPS, the $1/2$-MMS guarantee for $2$-witness EF1 is tight, as demonstrated by the following example.
While the $4/7$-MMS guarantee for $3$-witness EF1 may not be tight, it matches the currently known guarantee for EFX allocations~\cite{ABM2018}.
Therefore, to obtain a better guarantee, one would first need to improve the guarantee for EFX.
We remark that \cite{ABM2018} adopts a definition of EFX that only requires the removal of goods with positive values, which neither implies nor is implied by 3-witness EF1.
Nevertheless, the hard instance establishing the $0.5914$-MMS upper bound satisfies EFX in our definition, and thus the bound applies to 3-witness EF1 as well.
\Cref{ex:3-witness-example} below improves the bound to $10/17$ for 3-witness EF1 allocations.

\begin{example}[Upper bound of $1/2$-MMS on $2$-witness EF1]
    Consider an instance with $n$ agents and $3n - 2$ goods with identical and additive valuation $u$.
    Define $\eps \coloneqq 1 / (n + 1)$. The allocation and the value of each good are defined as follows:
    \begin{itemize}
        \item bundle $A_1$ contains a good with value $1$,
        \item for $2 \le i \le n$, bundle $A_i$ contains three goods with values $\eps, 1 - \eps, 1 - \eps$ respectively.
    \end{itemize}
    
    Since $A_1$ has the minimum value, and for each $2 \le i \le n$, removing either of the last two goods from $A_i$ results in a value at most $u(A_1)$, this allocation is $2$-witness EF1.
    On the other hand, we have $u(M) / n = 2 - 2 \eps$. 
    Since transferring a good of value $\eps$ from each $A_i$ to $A_1$ makes all bundles have value $u(M) / n$, we have $\mu^n(M) = 2 - 2 \eps$.
    Therefore, $u(A_1) = \frac{1}{2 - 2 \eps} \mu^n(M) = \frac{n + 1}{2n} \mu^n(M)$, and taking $n \to \infty$ demonstrates that the $1/2$-MMS bound in \Cref{thm:2-witness-implies-half-mms} is tight.
\end{example}

\begin{example}[Upper bound of $10/17$-MMS on $3$-witness EF1]
\label{ex:3-witness-example}    
    Consider an instance with $6n + 1$ agents with identical and additive valuation $u$.
    Define $\eps \coloneqq \frac{7}{30n}$. The allocation and the value of each good are defined as follows:
    \begin{itemize}
        \item bundle $A_1$ contains a good with value $1 + \eps$,
        \item for $2 \le i \le n + 1$, bundle $A_i$ contains seven goods with values $\eps, 0.2, 0.2, 0.2, 0.2, 0.2, 0.2$ respectively,
        \item for $n + 2 \le i \le 3n + 1$, bundle $A_i$ contains four goods with values $\eps, 0.5, 0.5, 0.5$ respectively,
        \item for $3n + 2 \le i \le 6n + 1$, bundle $A_i$ contains two goods with values $1 + \eps, 1 + \eps$ respectively.
    \end{itemize}
    
    Since $A_1$ has the minimum value, and for each $2 \le i \le 3n + 1$, removing any of the last three goods from $A_i$ results in a value at most $u(A_1)$, and for each $3n + 2 \le i \le 6n + 1$, removing any goods from $A_i$ results in a value at most $u(A_1)$, this allocation is $3$-witness EF1.
    On the other hand, we have $u(M) / (6n + 1) = 1.7 + \eps$.
    By assigning goods as follows, we can make all bundles have value $1.7 + \eps$, which implies $\mu^{6n + 1}(M) = 1.7 + \eps$:
    \begin{itemize}
        \item bundle $B_1$ contains a good with value $1 + \eps$ and $3n$ goods with value $\eps$ each,
        \item for $2 \le i \le 6n + 1$, bundle $B_i$ contains three goods with values $1 + \eps, 0.5, 0.2$ respectively.
    \end{itemize}
    Therefore, $u(A_1) = \frac{1 + \eps}{1.7 + \eps} \mu^{6n + 1}(M) = \frac{30n + 7}{51n + 7} \mu^{6n + 1}(M)$, and taking $n \to \infty$ demonstrates that $\frac{10}{17}$-MMS is an upper bound for \Cref{thm:3-witness-implies-47-mms}.
\end{example}


\section{Discussion}

In this paper, we studied fair division of indivisible goods in the comparison-based query model of Bu, Li, Liu, Song, and Tao~\cite{BLLST2024}.
Our main contribution was to propose a witness-certificate scaling framework to approach the optimal query complexity bound in the comparison-query model.

Despite our progress, several open questions remain.

\paragraph{EF1 for non-identical valuations.}
The main open problem is whether an EF1 allocation for non-identical additive valuations can be computed with $\mathrm{poly}(n, \log m)$ comparison queries.
The case $n \le 3$ is resolved with $O(\log m)$ queries~\cite{BLLST2024}, and the case $n = 4$ is already open.
Our partial result shows that witness information can partially substitute for cardinal information in restricted non-identical settings, but the general case remains open.
This open problem is not limited to the comparison-query model.
Even in Feige's stronger communication model, there still exists a gap between the description complexity of $O(n\log m)$ and the best known communication complexity $O(m\log m)$ (via round robin)~\cite{Fei25}.
Thus, even in this stronger setting, a non-identical EF1 allocation with the sublinear-in-$m$ communication cost is open.

\paragraph{Lower bounds and optimality.}
For identical valuations, the remaining gap is the $O(\log n)$ factor between our upper bound and the known lower bound.
Appendix~\ref{app:witness-repair-lb} suggests that closing the gap may require either abandoning witness certificates or proving a matching $\Omega(n \log n \log(m/n))$ lower bound.
For non-identical valuations, the gap is larger.
The best lower bound for PROP1 remains $\Omega(n \log(m/n))$~\cite{Fei25}, against our $O(n^3 \log m)$ upper bound, and it is open whether the dependence on $n$ can be pushed closer to quadratic or whether structural obstacles prevent such improvements.
We also note that for $\alpha$-MMS with $\alpha < 1$, the best known lower bound is only $\Omega(\log(m/n))$~\cite{BLLST2024}, as Feige's linear-in-$n$ bound is known only for exact MMS and PROP1~\cite{Fei25}.
Establishing lower bounds with nontrivial dependence on $n$ for approximate MMS is an interesting direction.

\paragraph{Beyond additivity.}
Our PROP1 results for non-identical valuations rely on additivity.
Extending accelerated frameworks to broader classes (e.g., monotone, submodular, unit demand) is a natural direction, potentially requiring new certification primitives expressible via comparisons.

\paragraph{Computing $k$-witness EF1 efficiently.}
It is open whether the scaling framework of \Cref{thm:sendmaxtomin} can be strengthened to output a $2$- or $3$-witness EF1 allocation within $O(\mathrm{poly}(n, \log m))$ queries.
In particular, once the bound becomes $O(n \cdot \mathrm{poly}(\log n, \log m))$ for $k=3$, \Cref{thm:2-witness-implies-half-mms,thm:3-witness-implies-47-mms} implies that the $1/2$-MMS guarantee of \Cref{thm:ef1-halfmms} is upgraded to $4/7$-MMS with almost no extra cost.
Dealing with two or more witnesses under transfers seems to require ordering information among the top goods of every bundle, which may require a linear number of queries in $m$.
 
\paragraph{Existence of $k$-witness EF1 beyond identical valuations.}
For identical additive valuations, EFX allocations exist~\cite{PR2020}, and hence $k$-witness EF1 allocations exist for every $k$.
For general additive valuations, EFX existence for three agents~\cite{CGM24} implies $k$-witness EF1 for every $k$ when $n \le 3$.
Thus, already for $n = 4$, even the existence of a $2$-witness EF1 allocation is open.
As the existence of an EFX allocation is a major open problem~\cite{Ama2023}, resolving existence for small $k$ may serve as a tractable stepping stone toward EFX.

\appendix

\section{A Sorting Barrier for Witness-Preserving Repair}
\label{app:witness-repair-lb}

In this section, we discuss a lower bound on algorithms that preserve witnesses for EF1.
Such algorithms contain an inherent sorting-type bottleneck.
We remark that this does not establish a lower bound for any EF1 allocation algorithms.
We assume that the agents have an identical additive valuation $u$.

To show the bottleneck, we consider the following witness-extension problem.
Suppose that we are given $n$ bundles $B_1,\dots,B_n$ and $n$ fresh goods $g_1,\dots,g_n$.
We are asked to find a bijection $\sigma\colon [n]\to[n]$ such that $(A_1, \dots, A_n)$ given by $A_i=B_i + g_{\sigma(i)}$ ($i\in [n]$) is 1-witness EF1 such that each $g_{\sigma(i)}$ is the witness for $A_i$, i.e., $u(B_i + g_{\sigma(i)}) \ge u(B_j)$ for all $i,j\in [n]$.
We say such a bijection $\sigma$ is feasible if $u(A_i)\ge u(A_j\setminus\{g_{\sigma(j)}\})=u(B_j)$ for all $i,j \in [n]$.

\begin{proposition}
Any deterministic comparison-query algorithm for the witness-extension problem requires $ \Omega(n\log n)$ queries in the worst case, even for identical additive valuations.
The same lower bound holds for randomized algorithms with constant success probability.
\end{proposition}
\begin{proof}
Fix any $\varepsilon\in(0,1)$ and a number $H>n$.
Suppose that the value of each given bundle $B_i$ is $u(B_i)=H-i$.
Hence, $\max_{i\in[n]} u(B_i)=H-1$.
The fresh goods have values $\varepsilon, 1+\varepsilon, 2+\varepsilon, \dots, n-1+\varepsilon$, but their order is hidden.
More precisely, for a hidden permutation $\pi$ over $[n]$, we set $u(g_j)=\pi(j)-1+\varepsilon$ for $j\in [n]$.
We assume a stronger condition that the algorithm is given the values of the bundles and the set of values of the fresh goods, and the only hidden information is the permutation $\pi$.
A lower bound under this condition applies to the comparison-query model.

We first show that every feasible bijection uniquely determines $\pi$.
Let $\sigma$ be a feasible bijection.
Since the fresh goods serve as witnesses, we must have $u(B_i)+u(g_{\sigma(i)})\ge u(B_1) = H-1$ for every $i\in[n]$.
Substituting the values implies $H-i+\bigl(\pi(\sigma(i))-1+\varepsilon\bigr)\ge H-1$, which means $\pi(\sigma(i))+\varepsilon\ge i$.
Since $\pi(\sigma(i))$ is an integer and $\varepsilon<1$, this implies $\pi(\sigma(i))\ge i$ for all $i\in [n]$.
As $\sigma$ is a bijection, the sequence $\pi(\sigma(1)),\pi(\sigma(2)),\dots,\pi(\sigma(n))$ is itself a permutation of $[n]$.
The only permutation $(r_1,\ldots,r_n)$ of $[n]$ satisfying $r_i\ge i$ for every $i$ is $r_i=i$ for every $i$.
Hence, $\pi(\sigma(i))=i$ for all $i\in [n]$.
Thus, $\sigma$ is feasible if and only if it maps each $B_i$ to the unique fresh good of value $i-1+\varepsilon$.
In particular, each hidden permutation $\pi$ has a unique feasible output $\sigma$.

Now consider any deterministic algorithm that makes at most $q$ comparison queries.
Its computation can be represented by a binary decision tree with at most $2^q$ leaves.
Since there are $n!$ possibilities of the hidden permutation and a leaf can output a feasible solution for at most one permutation, correctness on all inputs requires $2^q\ge n!$.
Therefore, the algorithm uses $q\ge \log_2(n!)=\Omega(n\log n)$ queries.

For randomized algorithms, take the distribution that chooses $\pi$ uniformly from all the permutations over $[n]$.
The same decision-tree argument shows that any deterministic $q$-query algorithm succeeds with probability at most $2^q/n!$ under this distribution.
Hence, if a randomized algorithm succeeds with probability at least $2/3$ on every input, then by Yao's minimax principle we must have $\frac{2^q}{n!}\ge \frac23$, and therefore, $ q\ge \log_2(n!)-O(1)=\Omega(n\log n)$.
\end{proof}

\section{Omitted Proofs}

\subsection{Section 3}\label{app:proof-section-3}

\begin{proof}[Proof of \Cref{prop:ef1-undecidable-ties}]
    Consider two agents and three goods $a,b,c$, with the fixed allocation $A_1=\{a\}$ and $A_2=\{b,c\}$.
    The two agents have identical additive valuations.
    We consider two scenarios: their valuation is $v$ given by
    $v(a)=1, v(b)=1, v(c)=2$, or $v'$ given by $v'(a)=1, v'(b)=1+\eps, v'(c)=2$ for a fixed $\eps\in(0,1)$.
    The allocation $\A$ is EF1 under $v$: agent 1's envy toward $A_2$ is eliminated by removing $c$, and agent 2 does not envy agent 1.
    On the other hand, $\A$ is not EF1 under $v'$, because agent 1's envy cannot be eliminated by removing either $b$ or $c$ from $A_2$.

    We observe that whenever $v(X)<v(Y)$, we also have $v'(X)<v'(Y)$ for any $X,Y$, since the bundle values under $v$ are integers and the perturbation changes any bundle value by at most $\eps<1$.
    The comparison for $v$ and $v'$ can differ only on $X,Y$ with $v(X)=v(Y)$.
    However, all such ties can be broken adversarially  and consistently with the strict order induced by $v'$.
    Therefore, any deterministic algorithm receives the same answers in both instances and must return the same decision on $\A$. 
\end{proof}

\subsection{Section 4}\label{app:non-identical-EF1-proof}

Here we prove \Cref{prop:non-identical-EF1-heavy-anchors,prop:non-identical-EF1-binary} based on the following lemma.

\begin{lemma}\label{lem:anchor-core-compression}
    Suppose that we are given an integer $L\ge 1$ and a set $H=\{h_1,\dots, h_n\}$ of pairwise distinct goods satisfying $u_i(M\setminus H)\le L\cdot u_i(h_i)$ for every $i\in [n]$.
    Then one can compute a 1-witness EF1 allocation using $O(n^2L\log m+\mathrm{poly}(n,L))$ comparison queries.
\end{lemma}

\begin{proof}
    We maintain a set $U$ of core goods and a matching from agents to core goods.
    Initially, $U=H$ and agent $i$ is matched to $a_i \coloneqq h_i$.
    We say that $a_i$ is an anchor good for agent $i$.
    For each agent, we keep the order of the current core goods in non-increasing order of value by using the comparison queries.

    Whenever a new good $h'$ is added to $U$, we update the matching as follows.
    While there exists an unmatched good $h'\in U\setminus\{a_1,\ldots,a_n\}$ and an agent $i$ such that $h'$ is ranked above $a_i$ in agent $i$'s order, we match agent $i$ to good $h'$ and make the old anchor $a_i$ unmatched.
    This process terminates because every such update strictly improves the rank of the updated agent.
    Hence, we have 
    \begin{equation}
    \label{eq:unmatched-core-safe-doubling}
        u_i(h')\le u_i(a_i)
        \qquad
        \text{for all } i\in[n] \text{ and } h'\in U\setminus\{a_1,\ldots,a_n\}.
    \end{equation}
    In particular, the value of each agent's anchor never decreases.

    We scan the goods of $R \coloneqq M\setminus H$ in an arbitrary fixed order.
    Let $N' \coloneqq [n]$.
    At any point, a bundle $B$ of goods is called safe if $\Comp_i(B,\{a_i\})=\true$ for every $i\in N'$.
    Let $\{x\}$ be the first remaining good.
    \begin{itemize}
        \item If $\{x\}$ is not safe:
        We have $u_i(x)\ge u_i(a_i)$ for some agent $i$.
        If this is the $(L+1)$st charge to such an agent $i$, then we set $N' \leftarrow N'\setminus \{i\}$ and retry this step.
        Otherwise, we add $x$ to the core $U$, remove $x$ from $R$, update the anchors, and continue.
        \item Otherwise, i.e., if $\{x\}$ is safe:
        Using the binary search, we find a safe bundle $B$ such that either $B$ contains all remaining goods, or $B + y$ is not safe for the next good $y$.
        If $B$ contains all remaining goods, then we output $B$ as a block, remove it from $R$, and terminate the scan.
        Otherwise, let $i\in N'$ be an agent witnessing that $B+y$ is not safe.
        If this is the $(2L+1)$st block cut by agent $i$, then we set $N' \leftarrow N'\setminus \{i\}$ and retry this step.
        Otherwise, we output $B$ as a block and remove it from $R$.
    \end{itemize}
    Since $u_i(a_i)$ only increases over time, every output block remains safe with respect to the anchors at any point.

    By construction, the numbers of core additions and blocks charged to agent $i$ are at most $L$ and $2L$, respectively.
    Thus, the total number of core additions and blocks is at most $O(nL)$.
    
    It remains to show that $u_i(B) \le u_i(a_i)$ for every block $B$ and every agent in $[n]$.
    Since the values of anchors never decrease, the claim holds for any block $B$ and agents in $N'$ at the time when $B$ is created.
    Fix agent $i$.
    We show that agent $i$ remains in $N'$ throughout if $u_i(h_i)>0$.
    \begin{itemize}
        \item Since $u_i(M\setminus H) \le L \cdot u_i(h_i)$, there exist at most $L$ distinct goods with value at least $u_i(h_i)$, and agent $i$ does not reach the $(L+1)$st core charge.
        \item Let $B_1, \dots, B_\alpha$ be candidate blocks charged to agent $i$.
        Let $x'_\ell$ be the first good after $B_\ell$.
        By the choice, $u_i(B_\ell + x'_\ell)\ge u_i(a'_i) \ge u_i(h_i)$, where $a'_i$ is the anchor when $B_\ell$ is charged to agent $i$.
        The blocks $B_1, \dots, B_\alpha$ are disjoint, and the next goods $x'_1,\dots, x'_\alpha$ are distinct.
        Thus, $\sum_{\ell} u_i(B_\ell + x'_\ell) \le 2u_i(M\setminus H) \le 2Lu_i(h_i)$.
        These observations imply that agent $i$ does not reach the $(2L+1)$st block charge.
    \end{itemize}
    Thus, agent $i$ is excluded from $N'$ only if $u_i(h_i)=0$.
    In this case, we have $0 \le u_i(M\setminus H) \leq L\cdot u_i(h_i)=0$.
    Since every block $B$ is a subset of $M\setminus H$, we have $u_i(B) \le u_i(h_i)$.

    Let $F=\{a_1,\ldots,a_n\}$ be the final anchor set.
    At this point, the set of meta-goods consists of the $O(nL)$ blocks and the $O(nL)$ unmatched core goods $U\setminus F$.
    We compute an EF1 allocation $(Y_1,\ldots,Y_n)$ of these meta-goods using the standard round-robin algorithm.
    Finally, we set $A_i\coloneqq Y_i + a_i$ for each $i\in [n]$ so that $\last(A_i) = a_i$.

    We show that $\A=(A_1,\ldots,A_n)$ is 1-witness EF1.
    Fix agents $i,j$.
    If $Y_j=\emptyset$, then $A_j-a_j=\emptyset$ and we are done.
    Otherwise, there exists a meta-good $Z\subseteq Y_j$ such that $u_i(Y_i)\ge u_i(Y_j\setminus Z)$.
    By \eqref{eq:unmatched-core-safe-doubling} and the condition of safe blocks, we have $u_i(Z) \le u_i(a_i)$.
    Thus, 
    \begin{align*}
        u_i(A_j-\last(A_j)) =u_i(Y_j)
        =u_i(Y_j\setminus Z)+u_i(Z)
        \le u_i(Y_i)+u_i(a_i)
        =u_i(A_i).
    \end{align*}
    Therefore, the allocation is 1-witness EF1.

    We analyze the query bound.
    By construction, $O(nL)$ core goods and blocks are created.
    Every time a core good is added, updating the anchor goods takes $O(n \log m+n^2)$ comparison queries.
    Cutting a block uses $O(n\log m)$ comparisons.
    The round-robin algorithm uses only $O(n^2L\log(nL))$ further comparison queries.
    Therefore, we use $O(n^2L\log m + n^3L+n^2L\log L)$ comparison queries in total.
\end{proof}

\begin{proof}[Proof of \Cref{prop:non-identical-EF1-heavy-anchors}]
    Since $P$ is unknown, we use a doubling scheme as follows.
    For $t=1,2,\ldots$, set $q_t=2^t n+1$.
    For each agent $i$, run \Cref{alg:ef1-halfmms} with $q_t$ bundles for valuation $u_i$.
    Let $T_i$ be a bundle with minimum value and let $W_i$ be the set of designated witnesses for EF1 for the agent $i$.
    Define $C_i=\{g\in W_i \mid \Comp_i(T_i,\{g\})=\true\}$.
    If the bipartite graph with a vertex set $[n] \cup \bigcup_{i\in [n]}C_i$ and an edge set $\{\{i,g\} \mid g\in C_i\}$ has no matching saturating all agents, we move on to $q_{t+1}$.    
    Once an agent-saturating matching exists, for each agent $i$, we assign edge weights consistent with the non-increasing order of values for the agent.
    This uses $O(\sum_i |C_i|\log |C_i|)=O(nq_t\log q_t)$ comparison queries.
    Among all agent-saturating matchings, we choose one of maximum total weight, which requires no further comparison queries.
    Let $a_i$ be the good assigned to agent $i$, and let $F\coloneqq\{a_1,\ldots,a_n\}$.
    By maximality, $u_i(a_i)\ge u_i(g)$ for every $g\in C_i\setminus F$.
    
    We claim that these goods satisfy the assumption of \Cref{lem:anchor-core-compression} with $L=2q_t$.
    Since $a_i\in C_i$, we have $u_i(a_i)\ge u_i(T_i)$.
    Consider any good $g\notin F$.
    If $g\in W_i\setminus C_i$, then $\Comp_i(T_i,\{g\})=\false$, and hence $u_i(g)\le u_i(T_i)\le u_i(a_i)$.
    If $g\in C_i\setminus F$, then $u_i(g)\le u_i(a_i)$ by the choice of $F$.
    Thus, in each of the $q_t$ bundles produced for agent $i$, the part other than its designated witness has value at most $u_i(T_i)\le u_i(a_i)$, and the witness itself, if not in $F$, also has value at most $u_i(a_i)$.
    Therefore, we have $u_i(M\setminus F)\le 2q_t \cdot u_i(a_i)$, and \Cref{lem:anchor-core-compression} applies with $L=2q_t$ and $H=F$.
    
    It remains to see that the doubling process to find $F$ stops.
    If the hidden goods $h_1,\ldots,h_n$ from the statement exist and $q_t \ge nP$, then $u_i(T_i)\le u_i(M)/q_t \le u_i(M)/(nP)< u_i(h_i)$.
    Hence, $h_i\in W_i$ by \Cref{lem:large-goods-in-witnesses}, and also $h_i\in C_i$.
    Since the $h_i$'s are pairwise distinct, a matching saturating all agents is obtained, and the process stops.
    At this point, $q_t = O(nP)$, i.e., $t=O(\log P)$.
    The doubling process iterates $O(\log P)$ steps, and each step uses $O(\mathrm{poly}(n, P)\log m)$ queries to compute $C_i$ for each $i$, and $O(\mathrm{poly}(n,P))$ queries to find $F$.
    By \Cref{lem:anchor-core-compression}, the query bound follows.
\end{proof}

\begin{proof}[Proof of \Cref{prop:non-identical-EF1-binary}]
    Let $q_t=2^t n+1$ for $t=1,2,\dots$.
    In each trial $t$, for each agent $i$, run $\ScalingSendMaxToMin$ with $q_t$ bundles for valuation $u_i$, and obtain a bundle $T_i$ with minimum value and a witness set $W_i$.
    Let $C_i = \{g\in W_i \mid \Comp_i(T_i, \{g\}) = \true\}$ for each agent $i$.
    If some $C_i$ is empty, abort this trial and move on to the next.

    We explain that this process stops once $q_t > K$.
    Since $\max_{i\in [n]} |S_i| \le K$, the bundle $T_i$ has value $0$ for every agent $i$.
    Since $|S_i|\ge 1$ ($i\in [n]$), \Cref{lem:large-goods-in-witnesses} implies that every good in $S_i$ is in $W_i$.
    Thus, $C_i\ne \emptyset$ ($i\in [n]$), and the process stops.

    Let $\tau$ be the trial when the above process stops.
    For each agent $i$, since $C_i \neq \emptyset$, we have $u_i(T_i) \le u_i(g) \le 1$ for some $g\in C_i$.
    This and \Cref{lem:scaling-sendmaxtomin-is-witness-EF1} imply that, in the result of $\ScalingSendMaxToMin$ for agent $i$, each nonempty bundle $B$ has value $u_i(B) = u_i(B - \last(B)) + u_i(\last(B)) \le u_i(T_i) + 1 \le 2$.
    Thus, we have $|S_i|=u_i(M) \le 2 q_{\tau}$ for each $i\in [n]$.
    Here, we recompute $C_i$ ($i\in [n]$) by the above process with $t=\tau+2$.
    Because $q_{\tau+2} > 2q_{\tau} \ge K$, we have $C_i \supseteq S_i$ for every $i\in [n]$.
    We remark that $q_{\tau+2} = O(n+K)$.
    
    Let each agent $i$ pick one best good $h_i$ from $(\bigcup_{j\in [n]} C_j)\setminus \{h_1,\dots, h_{i-1}\}$, where if no good remains, then agents pick arbitrary goods from $M\setminus(\bigcup_{j\in [n]} C_j \cup \{h_1,\ldots,h_{i-1}\})$ (since $m>n$, this set is nonempty).
    As $|C_i|\le q_{\tau+2}$ for every agent $i$, we have $\lvert\bigcup_i C_i\rvert\le nq_{\tau+2}$.
    Thus, the best remaining good for each agent can be found using $O(nq_{\tau+2})$ comparisons, and all choices require $O(n^2q_{\tau+2})=O(n^2(n+K))$ comparison queries in total.
    If $u_i(h_i)=1$, then $u_i(M - \{h_1,\dots, h_n\}) \le u_i(M) \le 2q_{\tau} u_i(h_i)$.
    If $u_i(h_i)=0$, then all the goods in $S_i$ are picked by other agents, and hence we have $u_i(M-\{h_1,\dots, h_n\}) =0 = 2q_{\tau} u_i(h_i)$.
    Therefore, we can apply \Cref{lem:anchor-core-compression} with $L=q_{\tau+2}$ and $H=\{h_1,\dots,h_n\}$.
    As $q_{\tau+2} = O(n+K)$, we use $O(\mathrm{poly}(n,K) \log m)$ comparison queries until we fix $C_i$ ($i\in [n]$), and the query bound also follows.
\end{proof}

\section*{Acknowledgments}
The second author was partially supported by JST ERATO Grant Number JPMJER2301, and JSPS KAKENHI Grant Numbers JP21K17708, JP21H03397, JP25K00137, JP26K14718, and JP26K02867.

\printbibliography[title={References}]

\end{document}